\documentclass[11pt]{article}
\usepackage[nohead,margin=1in]{geometry}
\usepackage{palatino}
\usepackage{amssymb,amsmath,amsfonts,eurosym,geometry,ulem,graphicx,caption,color,setspace,comment,footmisc,caption,natbib,pdflscape,subcaption,array}
\usepackage[colorlinks = true,
            linkcolor = blue,
            urlcolor  = blue,
            citecolor = blue,
            anchorcolor = blue]{hyperref}
\usepackage{natbib}
\usepackage[inline]{enumitem}
\usepackage{float}
\usepackage{pgfplots}
\usepackage{tikz}
\usepgfplotslibrary{fillbetween}
\usepackage{graphicx}
\usepackage{caption}
\usepackage{subcaption} 
\usepackage{float}      
\usepackage{algorithm}
\usepackage{algpseudocode}
\usepackage{lipsum}

\DeclareCaptionFormat{algor}{%
  \hrulefill\par\offinterlineskip\vskip1pt%
    \textbf{#1#2}#3\offinterlineskip\hrulefill}
\DeclareCaptionStyle{algori}{singlelinecheck=off,format=algor,labelsep=space}
\usetikzlibrary{arrows.meta, positioning}

\newtheorem{corollary}{\color{blue}Corollary}
\newtheorem{proposition}{\color{blue}Proposition}
\newtheorem{lemma}{\color{blue}Lemma}[section]

\newtheorem{assumption}{\color{blue}Assumption}
\newtheorem{definition}{\color{blue}Definition}

\newenvironment{proof}[1][Proof]{\noindent\textbf{#1.} }{\ \rule{0.5em}{0.5em}}

\newcolumntype{L}[1]{>{\raggedright\let\newline\\arraybackslash\hspace{0pt}}m{#1}}
\newcolumntype{C}[1]{>{\centering\let\newline\\arraybackslash\hspace{0pt}}m{#1}}
\newcolumntype{R}[1]{>{\raggedleft\let\newline\\arraybackslash\hspace{0pt}}m{#1}}
\pgfplotsset{compat=1.18}
\DeclareMathOperator\supp{supp}

\begin{document}

\title{Persuasion in Lemons Markets\footnote{We are deeply grateful to Piotr Dworczak, Jeff Ely and Alessandro Pavan for their guidance. All errors remain our own.}}
\author{Andrea Di Giovan Paolo\footnote{Department of Economics, Northwestern University. Email: \href{mailto:andreadigiovanpaolo@u.northwestern.edu}{andreadigiovanpaolo@u.northwestern.edu}} \and Jose Higueras\footnote{Department of Economics, Northwestern University. Email: \href{mailto:josehiguerascorona2025@u.northwestern.edu}{josehiguerascorona2025@u.northwestern.edu}}}
\date{September 23, 2025}

\maketitle
 \vspace{-2em}
 \begin{center}
     \Large \textbf{Preliminary Draft}
 \end{center}
\vspace{2em}
\begin{abstract}
We study information disclosure in competitive markets with adverse selection. Sellers privately observe product quality, with higher quality entailing higher production costs, while buyers trade at the market-clearing price after observing a public signal. Because sellers’ participation in trade conveys information about quality, the designer faces endogenous constraints in the set of posteriors that she can induce. We reformulate the designer’s problem as a martingale optimal transport exercise with an additional condition that rules out further information transmission through sellers' participation decisions, and characterize the optimal signals. When the designer maximizes trade volume, the solution features negative-assortative matching of inefficient and efficient sellers. When the objective is a weighted combination of price and surplus, optimal signals preserve this structure as long as the weight on the price is high enough, otherwise they fully reveal low-quality types while pooling middle types with high-quality sellers.

\end{abstract}

\newpage

\section{Introduction}\label{sec:introduction}

In markets with adverse selection, it is often valuable for policymakers and intermediaries to promote market thickness or higher trade prices. This can be achieved either through direct intervention in the market (e.g., \citealp{tiroleOvercomingAdverseSelection2012}), or indirectly by controlling the information available to market participants, which is the focus of this paper. An important case where such intervention is desirable is when trade generates positive externalities, so that broader participation in the market is socially valuable, even if some transactions might be privately inefficient. For example, \cite{herkenhoff_how_2024} show that relaxing household credit constraints increases job-finding rates and improves worker–firm matching. In financial-stability policy stress-test disclosures affect both the severity of adverse selection and the extent of interbank risk-sharing, and fully informative signals can be suboptimal precisely because they undermine insurance incentives \citep{goldsteinStressTestsInformation2018}. Labor markets provide another example, where expanding participation by disadvantaged or riskier workers can generate social benefits beyond the surplus in the individual firm-worker match \citep{pallais_inefficient_2014}. A second set of circumstances that motivate interventions arises when the designer’s payoff depends directly on market outcomes such as the equilibrium price or, more generally, a convex combination of price and surplus. Online labor platforms, for example, observe rich worker performance histories and earn commissions as a fraction of wages; full disclosure of productivity could depress wages and reduce matches, lowering both participation and platform revenue. At the same time, platforms may wish to leave some rents to sellers to reduce the risk of migration to competing platforms. 



Motivated by these considerations, we analyze a general framework of public information disclosure in a competitive market with adverse selection à la \cite{akerlofMarketLemonsQuality1970}. Each seller privately observes the quality of her product and incurs a production cost that increases with quality\footnote{Similar techniques to the ones presented in this paper allow to study the case with decreasing costs, in which the market exhibits \textit{favorable} selection: the sellers most eager to trade are those with higher-quality products.}. A designer—who values either the weighted volume of trade or some combination of the market price and social surplus—commits in advance to a public signal that determines what information about quality is disclosed. After observing the signal, buyers decide whether to purchase at the market-clearing price. In the competitive equilibrium, the price must equal the expected quality of the product conditional on \textit{all} information available to market participants (see, e.g., \citealp{mas1995microeconomic}; \citealp{azevedoPerfectCompetitionMarkets2017a}). This includes not only the public signal but also the fact that only sellers whose production cost is below the price are willing to trade. Hence, by shaping buyers’ beliefs, each signal affects the price, which in turn determines which sellers participate, leading to additional updates in beliefs and price. The designer therefore faces an information-design problem with a fixed-point constraint that pins down the market-clearing price.

We show that this can be reformulated as a martingale optimal transport problem with a free marginal (see \citealp{kolotilinPersuasionMatchingOptimal2025}), with an additional constraint---which we call the ``prices-as-means constraint" (\ref{price})--- ruling out the double updating of beliefs described above (Proposition \ref{prop:obedience}). This, together with martingale and Bayes plausibility constraints, characterizes the set of implementable posteriors. The intuition for the new constraint is simple. Recall that the designer cares about the volume of trade, that is, she gets a payoff of $0$ whenever trade does not occur. If she were to induce any posterior mean $x$ different from the price, this would imply that some sellers in the market segment characterized by an average quality $x$ are not trading. The designer can achieve a (weakly) better outcome by fully revealing the types that are not trading. The price in the original segment would stay the same, and the average quality would now equal the price. Moreover, looking at the sellers' who are now being fully revealed, if they carry products whose value to the buyer is higher than the production cost, this would result in more trades taking place, increasing the designer's payoff. Otherwise, these sellers would not trade and the designer's payoff would remain unchanged. Thus, we can without loss of optimality focus on signals that induce prices as posteriors means whenever trade takes place, and avoid the complexities of dealing with the fixed point. 

We first analyze the benchmark in which the cost function partitions types into two contiguous intervals separated by a cutoff. Types below the cutoff are \textit{inefficient}, in the sense that their cost lies above the buyer’s valuation, while types above the cutoff are \textit{efficient}. In this setting, when the designer cares more about higher types trading, the optimal signal is unique and takes a ``reveal–pool" form (Proposition \ref{prop:alpha_increasing}). The lowest-quality sellers are fully revealed and do not trade, while every other inefficient type is paired with exactly one efficient type so that the resulting price equals the efficient type’s cost. The optimal matching is negative assortative: higher efficient types are paired with lower inefficient types. If the designer instead places greater weight on trade by lower types, under a regularity condition, the optimal signal takes a ``pool–reveal–pool" form: very low and very high (those closest to the cutoff) inefficient types trade, while some intermediate types are revealed. Whenever types are pooled, the signal follows the negative assortative structure described above, with pairs of inefficient and efficient types matched together to a mean equal to the cost of the efficient type (Proposition \ref{prop:alpha_convex}). 

Finally, when the designer’s objective is a convex combination of price and surplus, the negative-assortative benchmark remains uniquely optimal provided the weight on price is sufficiently high. As the designer places more weight on surplus, the set of optimal signals expands. A signal is optimal if and only if it fully reveals types below a cutoff, and arbitrarily pools all remaining inefficient types with efficient types, subject only to the condition that each pool’s mean equals the cost of its highest type (Proposition \ref{prop:revenue-surplus}). 

We consider two extensions (Propositions \ref{prop:extension_multiple_crossings} and Corollary \ref{coro:extension_gainsatthebottom}). If the cost function is such that gains from trade arise at the bottom of the quality distribution rather than the top, the uniquely optimal policy is full revelation. If instead the cost function partitions the type space in multiple disjoint efficient and inefficient regions, the analysis applies within each inefficient-efficient pair of regions: the optimal policy repeats the reveal–then–pool construction, yielding separate negative-assortative pools.

\paragraph{Related literature.}
Our paper contributes to the literature on information design (\citealp{rayo_optimal_2010}; \citealp{kamenica_bayesian_2011}), focusing on settings where receivers are exposed to additional information through subsequent interactions with privately informed players (\citealp{bergemannBayesCorrelatedEquilibrium2016}). Specifically, we consider information disclosure in a competitive lemons market, in which buyers can learn information about the quality of a product from sellers' decisions to trade at the market-clearing price. Technically, the problem we study belongs to the class of \textit{martingale optimal transport} problems with a free marginal (see \citealp{dworczakPersuasionDuality2024}; \citealp*{kolotilinPersuasionMatchingOptimal2025}). \cite{kolotilinPersuasionMatchingOptimal2025}, in particular, identify conditions under which signals that involve negative assortative matching are optimal. Our approach departs from this literature in two respects. First, our objective function does not satisfy the regularity conditions assumed in these papers. Second, we introduce an additional constraint that rules out double updating of beliefs (\ref{price}), and effectively imposes a support restriction on the optimal coupling. These differences mean that, even after re-framing our problem as an optimal transport exercise, we cannot use the main technical result in \cite{kolotilinPersuasionMatchingOptimal2025} (Theorem 1) to characterize the optimal signal.
\smallskip

\noindent This work is also related to the broad literature on disclosure and intervention in markets with adverse selection, going back to \citet{akerlofMarketLemonsQuality1970}. \cite{tiroleOvercomingAdverseSelection2012} studies how a regulator can ``jump start" a frozen market (that is, increase market thickness) by offering to buy potentially toxic assets from firms. The regulator is constrained by the fact that sellers, who are informed about the true value of their assets, strategically decide whether to self-select into the government program or the market. We also study a case in which the designer does not fully control the final allocation and has to interact with a competitive market, but focus instead on how public \textit{information} policies can increase market thickness when trade generates positive externalities. The closest paper to ours is \citet{goldsteinStressTestsInformation2018}, who study in a discrete-type setting how stress-test disclosures in financial markets affect both the severity of adverse selection and the scope for risk sharing. Our analysis builds on an optimal transport formulation, where complementary slackness under strong duality provides a sharp characterization of the support of the primal solution. This allows us to identify the set of candidate optimal signals and then construct dual multipliers that verify their optimality. The approach makes it possible to analyze, in a continuous-type framework, applications and market structures beyond the financial-stability context, including settings where the designer values a combination of market price and seller surplus, or when the cost structure is such that gains from trade are not monotonic in the seller's type.

Finally, our paper is related to the literature studying how different market segmentations or information structures affect prices and output (\citealp{bergemannLimitsPriceDiscrimination2015}; \citealp{kartik_lemonade_2025}). \citet{kartik_lemonade_2025}, in particular, analyze an adverse selection setting where the seller makes a take-it-or-leave-it offer to the buyer and characterize all possible outcomes as the information available to each side varies. We also study the effect of information in lemons markets. Our focus, however, is on identifying the information structure that maximizes market thickness and a combination of price and producer surplus. Moreover, we examine a different trade protocol, assuming that the price is determined competitively in equilibrium rather than by a monopolistic seller.

The rest of the paper is organized as follows. Section \ref{sec:model} introduces the model and notation. Section \ref{sec:optimaltransport} shows how the problem can be reduced to an optimal transport problem. Section \ref{sec:main_results} presents the main results, characterizing the optimal signal for the case where the designer maximizes weighted volume of trade and a convex combination of price and seller surplus. Section \ref{sec:extensions} discusses some extensions. Finally, Section \ref{sec:conclusion} concludes.

\section{Model}\label{sec:model}

\paragraph{Preliminaries.} In what follows, for any $x \in X$, we use $\delta_{x}$ to denote the Dirac measure concentrated on ${x}$. For any joint distribution $\pi \in \Delta(\Theta \times X)$, we let $\pi_{\theta} \in \Delta(X)$ denote the conditional distribution over $X$ given $\theta$; we define $\pi_{x} \in \Delta(\Theta)$ analogously. Finally, whenever we say that a joint distribution $\pi \in \Delta(\Theta \times X)$ is the \textit{unique} solution to some optimization problem with a marginal constraint over $\Theta$ given by $F$, we mean that for any other solution $\pi' \in \Delta(\Theta \times X)$ satisfying the same marginal constraint, we have $\pi_{\theta} = \pi'_{\theta}$ for $F$-almost every $\theta \in \Theta$.

\subsection{General Framework}

We study a standard competitive “lemons'' market à la \cite{akerlofMarketLemonsQuality1970}, with a unit mass of sellers and free entry of buyers. Each seller privately observes her type $\theta\in\Theta=[0,1]$, drawn from a distribution $F:\Theta \to [0,1]$ with strictly positive continuously differentiable density $f:\Theta \to (0,\infty)$. Producing an indivisible good of quality $\theta$ costs the seller $c(\theta)$, where $c:\Theta\to(0,\infty)$ is strictly positive and continuously differentiable. Buyers have unit demand. If trade occurs at price $p$, the buyer’s payoff is: 
\[
u_b=\theta-p,
\] 
and the seller’s payoff is: 
\[
u_s=p-c(\theta),
\] 
while both buyers and sellers have an outside option worth zero. Free entry of buyers pins down the \textit{competitive equilibrium price} $p^*$ through the fixed-point condition (see Definition 13.B.1 in \citealp{mas1995microeconomic}):
\begin{equation*}
    p^* = \mathbb{E}\bigl[\theta \mid c(\theta)\le p^*\bigr],
\end{equation*}
whenever the fixed point $p^*$ exists, trade occurs at price $p^*$. Otherwise, we assume market breakdown and set $p^* =0$. At price $p^*$, all types with $c(\theta) \leq p^*$ engage in trade.
The expected quality of traded products is then $\mathbb{E}\left[\theta \mid c(\theta) \leq p^*\right]$, which equals the price paid, $p^*$. Thus, consumer surplus is zero, which is the necessary condition for market clearing under free entry of buyers.

We introduce in this framework an information designer who observes data on $\theta$ and can commit to any \textit{public signal structure} $\sigma:\Theta\to\Delta(S)$, where the signal space $S$ is rich enough to ensure that the designer faces no exogenous constraints on communication. After a \textit{signal realization} $s \in S$, the competitive equilibrium price $p(s)$ is given by:
\begin{equation}\label{eq:fixed point constraint}
    p(s) = \mathbb{E}_{\sigma}\bigl[\theta \mid c(\theta)\le p(s),\,s\bigr], \tag{$EQ_s$}
\end{equation}
if one exists, and otherwise there is no trade after $s$, and $p(s)=0$. Here the expectation is taken with respect to the posterior over types induced by $\sigma$.

While the designer could potentially eliminate adverse selection in the market by providing a fully informative signal, the focus of our analysis is on the many situations where it is important to consider alternative objectives than simple efficiency. We illustrate this with four examples:
\begin{itemize}
    \item \textbf{Credit scoring and financial inclusion.} When policymakers value aggregate borrowing to expand financial inclusion, too fine-grained credit scores may result in excessively low credit access. 
    \item \textbf{Freelancing platforms.} Online labor platforms (e.g., Upwork) match freelancers with firms and often have more accurate information about worker productivity. Since these platforms earn a percentage of workers' wages, they have an incentive to manipulate the information available to prospective employers to induce higher wages.
    \item \textbf{Bank stress tests.} When regulators release results from stress tests, this can reduce adverse selection but also distort banks' incentives to insure each other against idiosyncratic risks \textit{ex-ante}. \citet{goldsteinStressTestsInformation2018} show that in such settings, welfare maximization entails maximizing the total probability of trade, weighted by a function of the type of the bank.
    \item \textbf{University grading.} A university may design its grading scheme to maximize the probability that its students are hired, without internalizing the overall effect on social welfare.
\end{itemize}

These examples motivate two classes of designer objectives. First, a \emph{weighted volume of trade} objective, given by:
\begin{equation}\label{eq:objective_trade}
    v_1(\theta,p) \;=\; \alpha(\theta) \cdot \,\mathbf{1}_{\{p\ge c(\theta)\}},
\end{equation}
where $\alpha:\Theta\to(0,\infty)$ is a positive and bounded function allowing us to place different weights on agent types, thus reflecting redistributive concerns. Second, a \emph{convex combination of price and producer surplus}\footnote{Under the assumption of a competitive market, consumer surplus is $0$ in the aggregate, thus the objective is equivalent to maximizing a combination of price and \textit{social} surplus.}:
\begin{equation}\label{eq:objective_price}
    v_2(\theta,p) \;=\;\bigl[p-(1-\beta)c(\theta)\bigr] \cdot \mathbf{1}_{\{p\ge c(\theta)\}},
\end{equation}
with $\beta\in[0,1]$ determining the relative weight on revenue versus efficiency. When $\beta=0$, the objective reduces to efficiency and full revelation is optimal. Our aim is instead to characterize the optimal information structure for any $\beta>0$.

We will maintain the following assumption throughout the paper, unless stated otherwise:
\begin{assumption}[Gains at the top]\label{assn:increasing_gainstop}
    The cost function $c:\Theta \to (0,\infty)$ is strictly increasing. Moreover, there exists a $\theta^* \in (0,1)$ such that $c(\theta^*)=\theta^*$, $c(\theta)>\theta$ for $\theta<\theta^*$ and $c(\theta)<\theta$ for $\theta>\theta^*$.
\end{assumption}
This case captures situations where only high-type sellers are willing to part with the object at a price equal to $\theta$. Under Assumption~\ref{assn:increasing_gainstop}, we will sometimes refer to types $\theta>\theta^*$ as \emph{efficient} and to types $\theta<\theta^*$ as \emph{inefficient}.

The assumption that $c$ is strictly increasing is economically motivated: it captures the presence of adverse selection in the market, i.e., sellers with lower types have lower reservation values and are thus more eager to trade. For example, in the canonical automobile market of \citet{akerlofMarketLemonsQuality1970}, a lower type corresponds to a car of lower quality and therefore to a lower willingness to accept from the seller.

Assumption \ref{assn:increasing_gainstop} also requires that $c$ intersects the $45^\circ$ line exactly once, and from above. We relax both assumptions in Section \ref{sec:extensions}. 

Importantly, this rules out situations in which there are \emph{gains from trade at the bottom} of the type distribution rather than \emph{gains at the top}, i.e., the case where there is still a unique intersection $\theta^*$, but types below $\theta^*$ are efficient ($\theta>c(\theta)$) while types above are inefficient ($\theta<c(\theta)$). We show in Corollary~\ref{coro:extension_gainsatthebottom} that when there are gains of trade at the bottom, the designer’s problem becomes trivial: it is optimal to fully reveal all sellers’ types regardless of her objective. Intuitively, pooling inefficient and efficient types lowers the average type below the inefficient one, so inefficient types can never be induced to trade.

Restricting to a single intersection simplifies exposition. As we show in Section~\ref{sec:extensions}, when $c$ intersects the $45^\circ$ line finitely many times, the optimal information structure strongly resembles the single-intersection case, but the analysis becomes more cumbersome because one must keep track of multiple intersection points.

\section{Information design problem}\label{sec:optimaltransport}
Since in competitive markets with adverse selection the agents' decision to engage in trade can convey information about the state of the world, the information designer is constrained in the set of posteriors she can generate. Importantly, the constraint is endogenous, as it depends on the induced equilibrium price, which in turn depends on the signal. We now show that the designer's problem reduces to a mean persuasion problem, with an additional condition ensuring that market participants cannot learn more information from the downstream market interactions. This greatly simplifies the problem, as it allows us to avoid explicitly taking into account the updating process leading to the fixed-point (\ref{eq:fixed point constraint}).

To see this, consider a simple example in which there are only two types of sellers, $\theta_L=0$ and $\theta_H=1$, with $c_L=1/8$ and $c_H=1/2$. There is a fraction $3/4$ of low types. Consider the signal $\sigma$ given by:
\begin{center}
\begin{tabular}{ c c c }
 & $s_1$ & $s_2$ \\
 $\theta_L$ & $7/9$ & $2/9$ \\ 
 $\theta_H$ & $1/3$ & $2/3$ \\  
\end{tabular}
\end{center}
This generates the posterior means $x_1=\mathbb{E}_{\sigma}[\theta \mid s_1]=\frac{1}{8}$ (with probability $\frac{2}{3}$) and $x_2=\mathbb{E}_{\sigma}[\theta \mid s_2]=1/2$ (with probability $\frac{1}{3}$). Upon observing $s_1$, the market internalizes the fact that only low-type sellers are willing to trade, leading to an additional update of beliefs to a posterior mean (and a price of) $x'_1=0$. When $s_2$ realizes, instead, both worker types are willing to work at a price of $x_2$. Thus there is no additional updating of beliefs and the market clears at a price of $x_2$. It is obvious in this simple example that inducing a mean $x \in [c_L,c_H)$ (leading to ``double updating") is wasteful, as it involves sacrificing some high types who end up not trading. We can therefore improve on this signal by taking the high types in the support of $\sigma(\cdot\,|\,s_1)$ and fully revealing them. This allows us to relabel the signal realizations that result in trade so that each is ``unbiased'', meaning it coincides with the posterior mean it induces. In particular, consider $\sigma'$ defined as follows:
\begin{center}
\begin{tabular}{ c c c c }
 & $x_1=0$ & $x_2=\frac{1}{2}$ & $x_3=1$ \\
 $\theta_L$ & $7/9$ & $2/9$ & $0$ \\ 
 $\theta_H$ & $0$ & $2/3$ & $1/3$ \\  
\end{tabular}
\end{center}
The new signal $\sigma'$ is unbiased and clearly improves on the previous one under both types of designer objectives (\ref{eq:objective_trade}) and (\ref{eq:objective_price}): the high types who were previously pooled in $s_1$ are now being revealed in $x_3$ and thus get to trade at a price of $1$, thus increasing the average price, volume of trade and surplus. When $x_1$ and $x_2$ realize the equilibrium price is still $0$ and $1/2$, respectively.

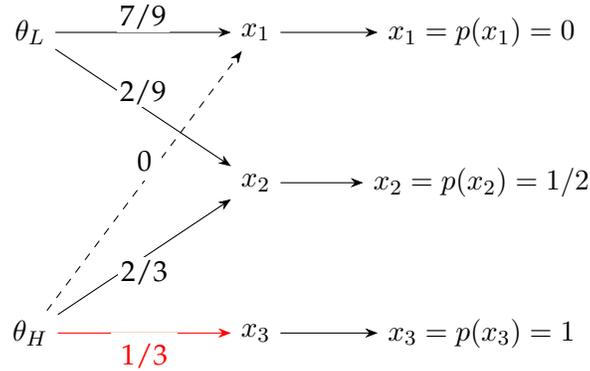
\begin{figure}[H]
  \centering
  \begin{tikzpicture}[>=Stealth]
    \node (A) at (0,4) {$\theta_L$};
    \node (D) at (0,0) {$\theta_H$};

    \node (B) at (3,4) {$x_1$};
    \node (C) at (3,2) {$x_2$};
    \node (E) at (3,0) {$x_3$};

    \node (F) at (6,4) {$x_1=p(x_1)=0$};
    \node (G) at (6,2) {$x_2=p(x_2)=1/2$};
    \node (H) at (6,0) {$x_3=p(x_3)=1$};

    \draw[->, inner sep=1pt] (A) -- node[above, fill=white] {7/9} (B);
    \draw[->, inner sep=1pt] (A) -- node[above, fill=white] {2/9} (C);
    \draw[dashed, ->] (D) -- node[above, fill=white] {0}   (B);
    \draw[->, inner sep=1pt] (D) -- node[below, fill=white] {2/3} (C);
    \draw[red,->] (D) -- node[below, fill=white] {1/3} (E);

    \draw[->] (B) -- (F);
    \draw[->] (C) -- (G);
    \draw[->] (E) -- (H);
  \end{tikzpicture}
  \caption{Construction of $\sigma'$ from $\sigma$}
  \label{fig:double_updating}
\end{figure}
The only additional technicality is that to fully formulate the problem as a mean persuasion problem we need to make sure that the fixed point always exists, even when the signal leads to market breakdown or no trade. To achieve this we define an auxiliary cost function:
\begin{equation*}\label{eq:costHat}
    \hat{c}(\theta)=\min\{\theta, c(\theta)\}.
\end{equation*}
See that whenever an inefficient type (i.e., one such that $\theta<c(\theta)$) is fully revealed, trade does not occur. However, under this auxiliary cost function the fixed-point condition is well defined as $\theta=\mathbb{E}_{\sigma}[\theta \mid\theta, \, \hat{c}(\theta) \leq \theta]$. 

As Figure \ref{fig:double_updating} illustrates, $\sigma'$ has two important properties: it is unbiased, and it eliminates double updating. Specifically, the posterior mean induced by $\sigma'$ coincides with the market-clearing price whenever there is no market breakdown. These observations are formalized in the following proposition:

\begin{proposition}[Prices as posterior means]\label{prop:obedience}
    Without loss of optimality, we can restrict attention to signal structures $\sigma$ that satisfy:
    \[
        x = \mathbb{E}_{\sigma}[\theta \mid x] = \mathbb{E}_{\sigma}[\theta \mid x,\; \hat{c}(\theta) \leq x].
    \]
\end{proposition}
We omit the proof of Proposition \ref{prop:obedience} as it trivially follows from the arguments sketched above.

The main advantage of Proposition \ref{prop:obedience} is that it allows us to formulate the designer's problem as an optimal transport problem, in the spirit of \citet{kolotilinPersuasionMatchingOptimal2023} and \citet{dworczakPersuasionDuality2024}, where instead of choosing arbitrary signal structures, the designer selects a joint distribution $\pi \in \Delta(\Theta \times X)$ over types $\Theta$ and posterior means $X \equiv [0,1]$ that are \textit{feasible}, i.e., that satisfy the following conditions:

\begin{align}
        \int_A \int_X d\pi(\theta, x) &= \int_A f(\theta)\,d\theta \quad \text{for all measurable } A \subseteq \Theta \tag{$BP$} \label{BP} \\
        \int_\Theta \int_B (x - \theta) \, d\pi(\theta, x) &= 0 \quad \text{for all measurable } B \subseteq X \tag{$M$} \label{martingale} \\
        \int_\Theta \int_B \mathbf{1}_{\{ \hat{c}(\theta) > x \}} \, d\pi(\theta, x) &= 0 \quad \text{for all measurable } B \subseteq X \tag{$PM$} \label{price}
    \end{align}
The first condition (\ref{BP}) is the standard \textit{Bayes' plausibility} constraint. The second is a \textit{martingale} constraint, requiring that, conditional on any $x$, the posterior distribution over types, $\pi_x$, has mean $x$. Finally, the last is an \textit{prices-as-means} constraint, which ensures that beliefs are not updated after the market observes the public signal $x$, implying that $x$ is also the price at which the market clears whenever trade occurs.

Moreover, observe that when we restrict attention to feasible joint distributions, any posterior mean $x < \theta^*$ necessarily leads to market breakdown. To see this, consider for instance the mean $x' < \theta^*$ in Figure~\ref{fig:example_increasing_cost}. By the martingale condition~(\ref{martingale}), any pooling of types that produces a posterior mean of $x'$ must include some $\theta > x'$, say $\theta''$. However, this violates the prices-as-means constraint~(\ref{price}), since $c(\theta'') > c(x') > x'$, where the first inequality follows from $c$ being increasing, and the second from $x' < \theta^*$.

\begin{figure}[H]
    \centering
    \includegraphics[width=0.55\linewidth]{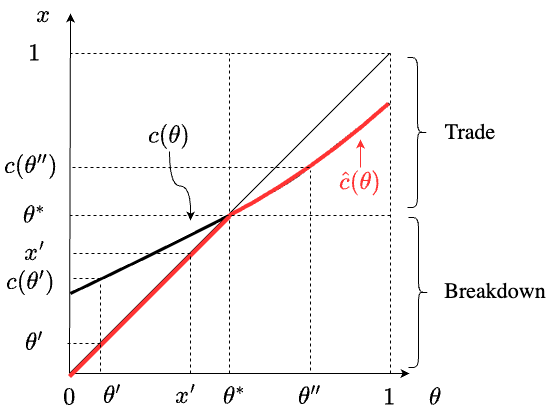}
    \caption{Example of cost function}
    \label{fig:example_increasing_cost}
\end{figure}

Thus, whenever a feasible joint distribution generates a posterior mean $x < \theta^*$, the only type that can be matched with $x$ is $\theta = x$, i.e.\ $\supp\,\pi_{x}=\{x\}$. Equivalently, type $\theta=x$ is revealed conditional on $x$.

Conversely, whenever a feasible joint distribution generates a mean $x \geq \theta^*$, trade occurs at price $x$. Indeed, the prices-as-means constraint~(\ref{price}) requires all types in the support of $\pi_x$ to satisfy $\hat{c}(\theta) \leq x$. If $\theta$ is efficient, then $c(\theta)=\hat{c}(\theta) \leq x$, so $\theta$ is willing to trade at price $x$. If $\theta$ is inefficient, then $c(\theta) < \theta^* \leq x$, and again $\theta$ is willing to trade at price $x$.

This observation is formalized in the following corollary: 

\begin{corollary}\label{cor:no_means_below_theta_star}
    Let $\pi$ be a feasible joint distribution. Then, for $\pi$-almost every $x < \theta^*$, we have $\supp\,\pi_x = \{x\}$ and market breakdown occurs. Conversely, for $\pi$-almost every $x \geq \theta^*$, trade occurs at price $x$. 
\end{corollary}

By Corollary \ref{cor:no_means_below_theta_star}, without loss of generality, we can write the designer’s objective in the weighted volume of trade case and in the convex combination of price and producer surplus as follows:
\begin{align*}
    v_1(\theta,x) \;&=\; \alpha(\theta) \cdot \,\mathbf{1}_{\{x\ge \theta^*\}},\\
    v_2(\theta,x) \;&=\;[x-(1-\beta)c(\theta)] \cdot \,\mathbf{1}_{\{x\ge \theta^*\}}.
\end{align*}

Then, the designer's problem can be written as:
\begin{align*}
         \max_{\pi \in \Delta(\Theta \times X)} \; &\int_0^1\int_0^1 v_i(\theta,x) \, d\pi(\theta, x) \tag{$P$} \label{primal}\\
         &\text{s.t \ref{BP}, \ref{martingale} and \ref{price}.}
\end{align*}
As noted in the introduction, from a technical point of view, this belongs to a class of problems known as \textit{martingale optimal transport}, but with the important difference that the target marginal measure is free (see, for instance, \citealp{kolotilinPersuasionMatchingOptimal2023} and \citealp{dworczakPersuasionDuality2024}). We impose an additional constraint, (\ref{price}), which ensures that any feasible $\pi$ does not generate a mean $x$ using types $\theta$ which end up not trading in the induced equilibrium. In other words, for every $\theta$ in the support of $\pi_x$ we have that $\hat{c}(\theta) \leq x$.

Our problem can also be interpreted as a multi-receiver game, in which one of the receivers (the buyer) is uninformed, and is simply trying to match the state, and the other receiver (the seller) is fully informed. The information asymmetry between the players implies that the uninformed agent can potentially garner information from the downstream interactions. According to this interpretation, the constraint (\ref{price}) is meant to capture obedience by the informed receiver, in the spirit of the Bayes correlated equilibrium of \cite{bergemannBayesCorrelatedEquilibrium2016}.

Additionally, to avoid trivial solutions to problem \ref{primal}, we impose the following assumption, which rules out feasible joint distributions achieving full trade:

\begin{assumption}[Infeasibility of full trade]\label{assn:no_full_trade}
    There is no feasible joint distribution $\pi \in \Delta(\Theta \times X)$ such that:
    \[
    \int_0^1\int_0^1\mathbf{1}_{\{x \geq \theta^*\}}\,d\pi(\theta,x)=1.
    \]
\end{assumption}

To derive most of our results we leverage duality. In particular, the dual problem to \ref{primal} is to find measurable functions $q:X\to \mathbb{R}$, $m:X\to \mathbb{R}$ and $w:\Theta\to \mathbb{R}$ that solve:
\begin{align*}
        \min_{w,q,p} \; &\int_0^1w(\theta)f(\theta)\;d\theta \tag{$D$} \label{dual} \\
        &\text{s.t} \quad w(\theta)\geq v_i(\theta,x)+q(x)(x-\theta)+m(x)\,\mathbf{1}_{\{ \hat{c}(\theta) > x \}}, \quad \text{for all $(\theta,x) \in \Theta \times X$.} \tag{$ZP$} \label{ZP}
\end{align*}
The interpretation of the dual formulation is as follows: $w(\theta)$ is the shadow price of type $\theta$, while $q(x)$ and $m(x)$ represent the values of relaxing the martingale constraint~(\ref{martingale}) and the prices-as-means constraint~(\ref{price}), respectively, at posterior mean $x$. Equation~\ref{ZP} then states that the shadow price of type $\theta$ is no less than the designer’s value from assigning type $\theta$ to any posterior mean $x$. This value consists of the designer’s objective $v_i(x,\theta)$, plus $q(x)$ multiplied by the degree to which constraint~(\ref{martingale}) is relaxed at $x$, and $m(x)$ multiplied by the degree to which constraint~(\ref{price}) is relaxed at $x$.

We say that there is \textit{no duality gap} if the value of the primal problem \ref{primal} equals the value of the dual problem \ref{dual}. We say that there is \textit{primal and dual attainment} if solutions exist for the primal and dual problems, respectively. Finally, we use the term \textit{strong duality} to refer to the case where there is both primal and dual attainment and no duality gap.

We will rely on the following two standard duality results in our analysis:

\begin{lemma}\label{cor:strong-duality}
     Suppose that the functions $(w, q, p)$ satisfy condition~\ref{ZP}, and that the joint distribution $\pi$ is feasible. Moreover, suppose that the value of the primal problem~(\ref{primal}) under $\pi$ equals the value of the dual problem~(\ref{dual}) under $(w, q, p)$. Then $\pi$ solves~\ref{primal}, $(w, q, p)$ solves~\ref{dual}, and strong duality holds.
\end{lemma}

\begin{lemma}[Complementary Slackness]\label{cor:complementary-slackness}
    Suppose that strong duality holds, and let $(w, q, v)$ be any solution to \ref{dual}. Then, for $F$-almost every $\theta \in \Theta$, we have:
    \[
        w(\theta) = \sup_{x \in X} \; v_i(\theta,x) + q(x)(x - \theta) + m(x)\, \mathbf{1}_{\{ \hat{c}(\theta) > x \}}.
    \]

    Moreover, a feasible $\pi$ solves~\ref{primal} if and only if, for $\pi$-almost every $(\theta, x) \in \Theta \times X$, we have:
    \[
        w(\theta)= v_i(\theta,x) + q(x)(x - \theta) + m(x)\, \mathbf{1}_{\{ \hat{c}(\theta) > x \}}.
    \]
\end{lemma}

To conclude this section, we introduce two examples of feasible joint distributions that will be central in the analysis: the \emph{reveal–pool} signal and the \emph{pool–reveal–pool} signal. 

\paragraph{Reveal–Pool Signal ($\pi^{\text{NAM}}$).} 
The \emph{reveal–pool} signal fully reveals inefficient types up to a cutoff $\underline{\theta}$. 
Each of the remaining inefficient types in $[\underline{\theta},\theta^*]$ is paired with a corresponding efficient type above $\theta^*$, so that the average quality of the pair exactly equals the production cost of the efficient type. 
The matching is \emph{negative assortative}: the first inefficient type to be pooled, $\underline{\theta}$, is matched to the highest efficient type, $1$; the next inefficient type is paired with the next-highest efficient type, and so on, until the cutoff $\theta^*$. 
At that point, type $\theta^*$ is revealed, (i.e., matched to itself).  

\paragraph{Pool–Reveal–Pool Signal ($\pi^{x^*-\text{NAM}}$).} 
The \emph{pool–reveal–pool} signal begins by pooling the very lowest inefficient types with efficient ones. 
Starting from type $0$, each inefficient type in $[0,\theta_1]$ is paired with an efficient type above $\theta^*$ so that their average quality equals the efficient type’s cost. 
After this initial pooling, there is an intermediate region $(\theta_1,\theta_2)$ of inefficient types who are fully revealed and therefore do not trade. Pooling then resumes: types in $[\theta_2,\theta^*]$ are again matched in a negative assortative way with efficient types, until reaching the cutoff $\theta^*$. 
As before, type $\theta^*$ is revealed.   

Figure~\ref{fig:nam_tnam} shows an example of the two types of signals.

\begin{figure}[H]
    \centering
    \begin{subfigure}[t]{0.48\textwidth}
        \centering
        \includegraphics[width=\linewidth]{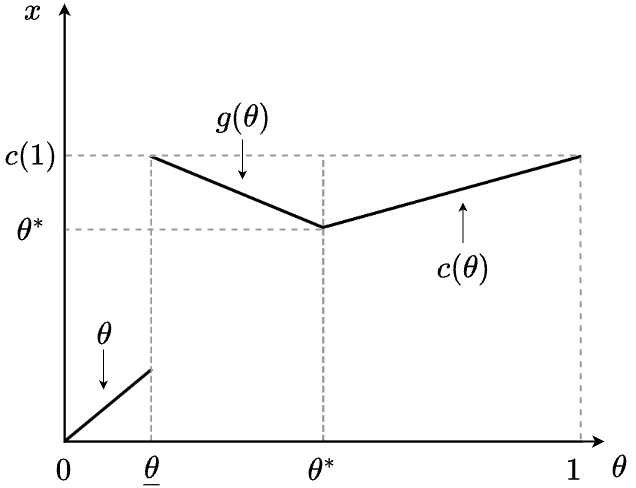}
        \caption{Reveal–pool signal ($\pi^{\text{NAM}}$).}
        \label{fig:nam}
    \end{subfigure}
    \hfill
    \begin{subfigure}[t]{0.48\textwidth}
        \centering
        \includegraphics[width=\linewidth]{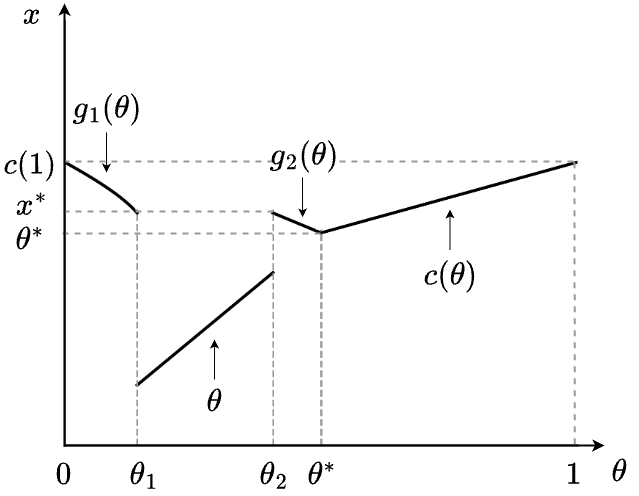}
        \caption{Pool–reveal–pool signal ($\pi^{x^*-\text{NAM}}$).}
        \label{fig:tnam}
    \end{subfigure}
    \caption{Two important feasible joint distributions: ``reveal–pool" and ``pool–reveal–pool".}
    \label{fig:nam_tnam}
\end{figure}


Note that in both signals all efficient types trade, and are matched to their cost of production. The shape of the functions $g$, $g_1$ and $g_2$ is determined by feasibility and in particular by the martingale condition: given a signal realization $x$, the buyer’s posterior mean must equal $x$. Looking for instance at $\pi^{\text{NAM}}$, any mean $x$ resulting in trade is generated by pairing an inefficient type $g^{-1}(x)$ with an efficient type $c^{-1}(x)$. Posterior weights are determined by Bayes’ rule and are proportional to the prior density scaled by the derivatives of $g$ and $c$, which capture how $\theta$-mass is stretched in $x$-space. 
For instance, under the (differentiable and monotonic) transformation $g$, a small mass around $\theta$, $f(\theta)d\theta$, gets mapped into some interval around $x=g(\theta)$ of length $dx=g'(\theta)d\theta$. Since mass is preserved under the transformation, the corresponding mass around $x$ then is given by
$$\frac{f(g^{-1}(x))}{|g'(g^{-1}(x))|}dx$$
At the same time, there will be an analogous contribution to the mass around $x$ by the efficient types, represented by
$$\frac{f(c^{-1}(x))}{|c'(c^{-1}(x))|}dx$$
The relative size of the two determines the weight given to the inefficient and efficient type, respectively. Intuitively, this means that since $c'$ is fixed, the crucial factor to guarantee feasibility is the slope of $g^{-1}$. If $g^{-1}$ is steep, a small set of inefficient types spreads thinly over many values of $x$, making their posterior weight negligible and pushing the mean toward the efficient type $c^{-1}(x)$. If $g^{-1}$ is flat, many inefficient types are concentrated in the same $x$, making the mean closer to $g^{-1}(x)$. The posterior mean is a weighted average of $c^{-1}(x)$ and $g^{-1}(x)$, and the martingale condition requires it to equal exactly $x$:
\begin{equation}
        x=g^{-1}(x)\frac{-f\left(g^{-1}(x)\right)\frac{\partial g^{-1}(x)}{\partial x}}{-f\left(g^{-1}(x)\right)\frac{\partial g^{-1}(x)}{\partial x}+f\left(c^{-1}(x)\right)\frac{\partial c^{-1}(x)}{\partial x}}+c^{-1}(x)\frac{f\left(c^{-1}(x)\right)\frac{\partial c^{-1}(x)}{\partial x}}{-f\left(g^{-1}(x)\right)\frac{\partial g^{-1}(x)}{\partial x}+f\left(c^{-1}(x)\right)\frac{\partial c^{-1}(x)}{\partial x}} \tag{ODE} \label{ODE}
    \end{equation}
This same ODE uniquely determines the decreasing bijections $g$, $g_1$ and $g_2$ that we use in the construction of the feasible signals. The functions differ only in terms of the initial/terminal conditions. In particular, $g$ and $g_2$ must match $\theta^*$ with itself: $g^{-1}(\theta^*)=\theta^*$ and $g_2^{-1}(\theta^*)=\theta^*$, while $g_1$ matches $0$ to $c(1)$: $g_2^{-1}(0)=c(1)$. Since they solve the same terminal value problem, $g$ and $g_2$ are identical, except for the initial point in the domain, which is $\underline{\theta}$ for $g$ and $\theta_2$ for $g_2$. Indeed, the main difference between $\pi^{x^*-\text{NAM}}$ and $\pi^{\text{NAM}}$ is that $\pi^{x^*-\text{NAM}}$ begins matching inefficient types with efficient types in a negative assortative way starting from the lowest type, $0$. However, because full trade is impossible (Assumption~\ref{assn:no_full_trade}), it is not feasible to match all inefficient types in this way. The matching must stop at some point, perfectly revealing those inefficient types for which market breakdown occurs—namely, the types in $(\theta_1,\theta_2)$. After this, $\pi^{x^*-\text{NAM}}$ resumes matching in a negative assortative way.

The existence and uniqueness of $g$, $g_1$ and $g_2$ follow from our assumptions on the density $f$ and cost function $c$, and is proved formally in the Appendix (Lemma \ref{lem:bijective-function}).

We can now formally define $\pi^{\text{NAM}}$ and $\pi^{x^*-\text{NAM}}$.

\begin{definition}[Reveal-Pool Negative Assortative Matching]\label{def:nam}
The \textbf{reveal-pool negative assortative matching} joint distribution, denoted by $\pi^{\text{NAM}} \in \Delta(\Theta \times X)$, is defined by:
\[
    d\pi^{\text{NAM}}(\theta,x)=
    \begin{cases}
        f(\theta)\delta_{\theta}(x) \,d\theta & \text{if } \theta \in [0,\underline{\theta}), \\
        f(\theta)\delta_{g_2(\theta)}(x) \,d\theta & \text{if } \theta \in [\underline{\theta},\theta^*], \\
        f(\theta)\delta_{c(\theta)}(x) \,d\theta & \text{if } \theta \in (\theta^*,1].
    \end{cases}
\]
\end{definition}

\begin{definition}[Pool-Reveal-Pool Negative Assortative Matching]\label{def:tnam}
Let $x^* \in (\theta^*,c(1))$, and define $\theta_1 = g^{-1}_1(x^*)$ and $\theta_2 = g^{-1}_2(x^*)$, so that $0 < \theta_1 < \theta_2 < \theta^*$.  
The \emph{pool-reveal-pool negative assortative matching} joint distribution, denoted $\pi^{x^*-\text{NAM}} \in \Delta(\Theta \times X)$, is given by:
\[
    d\pi^{x^*-\text{NAM}}(\theta,x)=
    \begin{cases}
        f(\theta)\delta_{g_1(\theta)}(x) \,d\theta & \text{if } \theta \in [0,\theta_1], \\
        f(\theta)\delta_{\theta}(x) \,d\theta & \text{if } \theta \in (\theta_1,\theta_2), \\
         f(\theta)\delta_{g_2(\theta)}(x) \,d\theta & \text{if } \theta \in [\theta_2,\theta^*], \\
        f(\theta)\delta_{c(\theta)}(x) \,d\theta & \text{if } \theta \in (\theta^*,1].
    \end{cases}
\]
\end{definition}

The following lemma formalizes that both $\pi^{x^*-\text{NAM}}$ and $\pi^{\text{NAM}}$ are feasible joint distributions:

\begin{lemma}[Feasibility of $\pi^{\text{NAM}}$ and $\pi^{x^*-\text{NAM}}$]\label{lem:nam-feasible}
    The joint distributions $\pi^{\text{NAM}}$ and $\pi^{x^*-\text{NAM}}$ are feasible; that is, they satisfy conditions~(\ref{BP}),~(\ref{martingale}), and~(\ref{price}).
\end{lemma}
The proof of Lemma~\ref{lem:nam-feasible} is provided in Appendix~\ref{sec:preliminary}.

\section{Results}\label{sec:main_results}

\subsection{Volume of trade}\label{sec:volume_trade}

We now present the two main results of the paper, which characterize the optimal joint distribution solving the primal problem~(\ref{primal}) when the designer’s objective is weighted volume of trade, with $\alpha(\theta)$ denoting the weight attached to type $\theta$.

\begin{proposition}\label{prop:alpha_increasing}
    Suppose that, for every $x \in [\theta^*,c(1)]$, the mapping 
    $\theta \mapsto \frac{\alpha(\theta)}{x-\theta}$ is strictly increasing on $[0,\theta^*]$. Then $\pi^{NAM}$ is the unique solution to the primal problem~(\ref{primal}).
\end{proposition}

\begin{proposition}\label{prop:alpha_convex}
Suppose that, for every $x \in [\theta^*,c(1)]$, the mapping $\theta \mapsto \frac{\alpha(\theta)}{x-\theta}$ is strictly convex on $[0,\theta^*]$, and that, for $x = c(1)$,
\[
    \frac{\alpha(0)}{c(1)} > \frac{\alpha(\underline{\theta})}{c(1)-\underline{\theta}}.
\]
Then there exists $x \in (\theta^*,c(1))$ such that $\pi^{x-\text{NAM}}$ is a solution to the primal problem~(\ref{primal}).  

Conversely, if $\pi$ is a solution to the primal problem~(\ref{primal}), then there exists $x \in (\theta^*,c(1))$ such that $\pi_{\theta} = \pi^{x-\text{NAM}}_{\theta}$ for $F$-almost every $\theta$.
\end{proposition}

The proofs of Propositions~\ref{prop:alpha_increasing} and \ref{prop:alpha_convex} are in Appendix~\ref{sec:proof_volume}. Technically, both results follow from duality (Lemma~\ref{cor:strong-duality}) and complementary slackness (Lemma~\ref{cor:complementary-slackness}). The argument proceeds as follows. We assume strong duality, and by manipulating the complementary slackness conditions we derive sharp necessary conditions that any optimal joint distribution must satisfy. These conditions pin down a negative assortative matching structure. We then construct a feasible dual solution supporting this distribution, which confirms not only that negative assortative matching solves the problem, but also that strong duality holds---implying uniqueness of the solution.

We next turn to the economic intuition. The ratio $\tfrac{\alpha(\theta)}{x-\theta}$ governs the solution: the numerator measures the value the designer assigns to type $\theta$ trading, while the denominator reflects how much the martingale constraint is relaxed when $\theta$ is matched with $x$. The closer $\theta$ is to $x$, the smaller the denominator and the greater the relaxation. Intuitively, the designer therefore prefers trade involving types with higher ratios.

When the ratio is increasing, higher types are more valuable (this is the case when $\alpha$ itself is increasing). This leads to the reveal–pool structure in Proposition~\ref{prop:alpha_increasing}. As $\alpha$ becomes decreasing at faster rates, the ratio can be locally decreasing. Note that, however, it cannot be globally decreasing: as $x \downarrow \theta^*$, the denominator vanishes and the ratio explodes for types close to $\theta^*$. Proposition~\ref{prop:alpha_convex} therefore captures the complementary case, where the designer values trade with very low types. Here the ratio first decreases but eventually must rise, and convexity offers a tractable way to model this shape.

We now turn to a more detailed economic explanation of the shape of the solution.

\paragraph{Who to pool and who to reveal}
To build intuition for the form of the optimal signal, recall that there are three forces at play. First, the payoff: the designer obtains a value of $\alpha(\theta)$ whenever type $\theta$ trades. Second, the martingale condition, which requires that prices equal the average type in any pool. Third, the prices-as-means constraint: every seller in a pool must be willing to trade at the induced price.  

Because of the constraint (\ref{price}), any pool containing inefficient types (those below $\theta^*$) must also include some efficient types (those above $\theta^*$). A pool with only inefficient types would yield means below the threshold $\theta^*$, thereby violating (\ref{price}). Moreover, every efficient type must be pooled; otherwise we would be wasting resources, since it could be matched with an inefficient type to increase trade probability. For the same reason, each pool must achieve a mean exactly equal to the cost of its highest type. If the mean were strictly higher, more inefficient types could be included in trade, without violating (\ref{price}). Since costs are strictly increasing, this implies that each pool must contain exactly one efficient type.  These simple observations tell us that the optimal signal matches each efficient type $\theta>\theta^*$ with a mean equal to its cost $c(\theta)$.  

We must then establish which \textit{inefficient} types should be allowed to trade. Let us consider the dual formulation of the problem and suppose a type $\theta' < \theta^*$ is pooled with some higher type in order to generate mean $x$. By complementary slackness, it must be strictly better to match $\theta'$ with $x$ than to reveal it (which would yield a payoff of $0$). This requires  
\[
\alpha(\theta')+q(x)(x-\theta')\;\geq\; 0.
\]  
Rearranging, we obtain  
\[
t_x(\theta') \equiv \frac{\alpha(\theta')}{x-\theta'} \;\geq\; -q(x).
\]  

The ratio $t_x(\theta)$ in left-hand side highlights how the payoff term $\alpha(\theta)$ and the martingale condition jointly determine the optimal solution. Including an inefficient type $\theta<\theta^*$ in a pool with mean $x\geq \theta^*$ generates a direct benefit of $\alpha(\theta)$ in the objective. But it also carries a cost: the farther $\theta$ lies below $x$, the harder it becomes to keep the average at $x$. This creates a bias in favor of inefficient types closer to $\theta^*$, since they are ``cheaper'' to include relative to the martingale condition. The overall shape of the optimal signal thus depends on how $\alpha(\theta)$ weights higher versus lower types. When $\alpha$ is increasing, for instance, the payoff and martingale forces are aligned, and the ratio $\frac{\alpha(\theta)}{x-\theta}$ is itself increasing. Hence, for any $\theta''>\theta'$,  
\[
\frac{\alpha(\theta'')}{x-\theta''}>-q(x) 
\quad \Rightarrow \quad 
\alpha(\theta'')+q(x)(x-\theta'')>0,
\]  
which implies that $\theta''$ should also be matched rather than revealed. This is show in Figure \ref{fig:alphaoverxtheta_increasing}.

\begin{figure}[H]
  \centering
  \begin{subfigure}{0.48\textwidth}
    \centering
    \includegraphics[width=\linewidth]{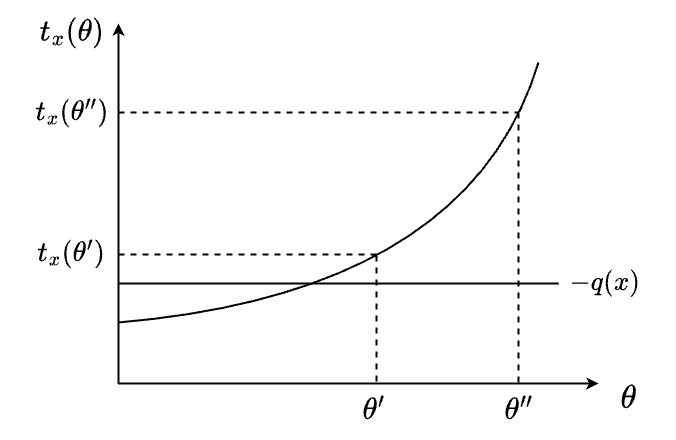}
    \caption{Increasing ratio $\alpha(\theta)/(x-\theta)$}
    \label{fig:alphaoverxtheta_increasing}
  \end{subfigure}
  \hfill
  \begin{subfigure}{0.47\textwidth}
    \centering
    \includegraphics[width=\linewidth]{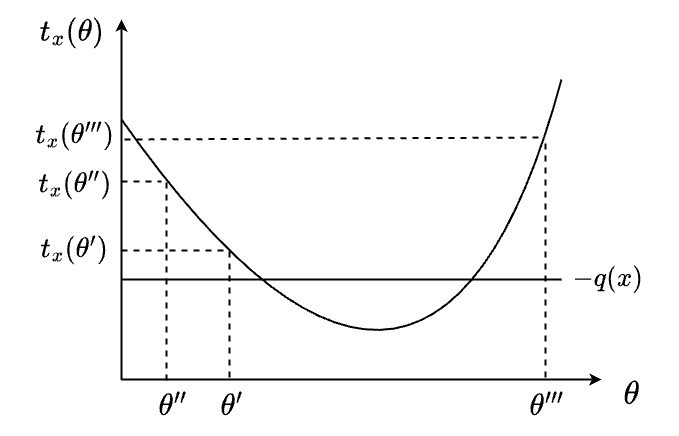}
    \caption{Convex ratio $\alpha(\theta)/(x-\theta)$}
    \label{fig:alphaoverxtheta_convex}
  \end{subfigure}
  \caption{Pooling intervals for different shapes of $\alpha(\theta)/(x-\theta)$.}
  \label{fig:alphaoverxtheta}
\end{figure}

By contrast, when $\alpha$ is decreasing, the two forces move in opposite directions: the designer prefers lower types, while the martingale constraint favors higher ones. Once again, it is the ratio $\frac{\alpha(\theta)}{x-\theta}$ that resolves the trade-off. If this ratio is strictly convex, then for very low types the payoff component dominates, so the designer pools them with efficient types. For intermediate types the martingale force dominates, so they are revealed. Finally, types just below $\theta^*$ are pooled again, since they are the cheapest inefficient types to induce to trade. Figure \ref{fig:alphaoverxtheta_convex} illustrates this case: if we pool some $\theta'$ to achieve a mean of $x$, then revealing $\theta''$ and $\theta'''$, for instance, would be sub-optimal, because it is strictly more profitable to match them with mean $x$.

The only thing left to determine is the shape of the optimal mapping in the pooling region for the inefficient types. We do this in two steps: we first show that the sign of the slope of the mapping is \textit{negative}, that is, the optimal signal displays negative assortative matching. This then allows us to use the martingale condition to pin down the exact shape of the function.

\paragraph{Negative assortative matching.} Intuitively, negative assortative matching is optimal because efficient sellers with a higher $\theta$ (those closer to $1$) have a comparative advantage in allowing \textit{lower} inefficient types (those closer to $0$) to trade. To see this, consider the discrete case 
$\Theta=\{\theta_{1},\dots,\theta_{N}\}$. For notational convenience, write 
$c_{i}\equiv c(\theta_{i})$ for the production cost of type $\theta_i$ and 
$f_{i}$ for its prior mass. Partition types into efficient 
$\mathcal H\equiv\{\theta_{j}: c_{j}<\theta_{j}\}$ and inefficient 
$\mathcal L\equiv\{\theta_{k}: c_{k}>\theta_{k}\}$. We use indices $j$ for efficient types 
and $k$ for inefficient types. 

From the considerations in the previous paragraph we know that the optimal signal assigns each efficient type $j$ to a pool 
containing itself and a mixture of inefficient types such that the posterior mean 
equals its cost. If $\pi_{jk}$ denotes the mass of inefficient type $k$ matched 
with efficient type $j$, then, the martingale condition requires
\begin{equation}\label{eq:martingale_resources}
    \underbrace{f_j(\theta_j-c_j)}_{ \text{resources available} }=\underbrace{\sum_{k\in\mathcal L}\pi_{jk}(c_j-\theta_k)}_{ \text{resources employed} },\qquad j\in\mathcal H,
\end{equation}
subject to feasibility $\sum_{j}\pi_{jk}\le f_k$ and $\pi_{jk}\ge 0$. Since all 
efficient types participate in trade, maximizing total trade is equivalent to 
maximizing $\sum_{j,k}\pi_{jk}$.

This setup can be interpreted as a production problem with $|\mathcal H|$ plants, corresponding to the efficient types,
each endowed with $f_j(\theta_j-c_j)$ units of resources. There are $|\mathcal L|$ products corresponding 
to the inefficient types, and producing one unit of product $k$ at plant $j$ 
requires $(c_j-\theta_k)$ units of resources. The planner’s objective is to 
maximize total output.

Now fix two plants $\theta_a<\theta_b$ and two products $\theta_1<\theta_2$. Suppose the 
matching has some degree of positive assortativeness, with 
$\pi_{b2}>0$ and $\pi_{a1}>0$. Consider the following swap: shift some mass 
from $\pi_{b2}$ to $\pi_{b1}$, adjusting within type $b$’s pool to preserve the 
martingale condition. Then, reallocate the corresponding amounts in type $a$’s 
pool so that feasibility is satisfied: add the mass removed from $\pi_{b2}$ to $\pi_{a2}$ and subtract the mass 
added to $\pi_{b1}$ from $\pi_{a1}$.

After a small increase in $\pi_{b1}$, denoted by $d\pi_{b1}$, the martingale condition must satisfy
\[
(c_b-\theta_{2})\,d\pi_{b2}+(c_b-\theta_{1})\,d\pi_{b1}=0
\quad\Rightarrow\quad
d\pi_{b2}=-\frac{c_b-\theta_{1}}{c_b-\theta_{2}}\,d\pi_{b1}.
\]
Thus, to marginally increase production of good $1$, plant $b$ must reduce production of 
type $2$ by
\[
\frac{c_b-\theta_{1}}{c_b-\theta_{2}},
\]
which is decreasing in $c_b$. Hence, it is relatively cheaper for higher-quality 
plants to produce lower-quality goods.

Turning to firm $a$, after using the fact that
$$d\pi_{a1}=-d\pi_{b1}, \ \ d\pi_{a2}=-d\pi_{b2}=\frac{c_b-\theta_1}{c_b-\theta_2},$$
the change in resource use (the right-hand side in Equation \ref{eq:martingale_resources}) is
\[
 -\Big[(c_a-\theta_{1})
-(c_a-\theta_{2})\frac{c_b-\theta_{1}}{c_b-\theta_{2}}\Big]d\pi_{b1}.
\]
Since 
\[
\frac{c_a-\theta_{1}}{c_a-\theta_{2}}>\frac{c_b-\theta_{1}}{c_b-\theta_{2}},
\]
the bracketed term is positive, implying that the total resource use in plant $a$’s 
pool decreases. This relaxation of the constraint makes it possible to accommodate 
more inefficient types while still satisfying feasibility and the martingale condition. Therefore, any allocation with positive assortativeness can be strictly improved 
by such swaps. 

In short, higher quality plants have a comparative advantage in producing lower quality products. Thus, to maximize total production, it is efficient that they focus exclusively on the production of those goods. 

Back to our persuasion setting, it follows that the optimal allocation must feature negative assortative matching: more efficient types are paired with less efficient ones. Thus, the mapping must be decreasing.  

\paragraph{The martingale condition.} The martingale condition further requires that, conditional on a signal $x$, the buyer’s posterior mean equals $x$. In this setting, each signal corresponds to a pool containing two types: the efficient type $c^{-1}(x)$ and the inefficient type $g^{-1}(x)$. Posterior weights are determined by Bayes’ rule and are proportional to the prior density scaled by the derivatives of $g$ and $c$, which capture how $\theta$-mass is stretched in $x$-space. As discussed in Section \ref{sec:optimaltransport} the feasible $g$ solves the differential equation (\ref{ODE}).

\subsection{Price and surplus}\label{sec:price_surplus}
We now characterize the class of optimal signals maximizing a convex combination of the trade price and producer surplus. Under the competitive market assumption, aggregate consumer surplus is zero, so this objective is equivalent to maximizing a combination of price and social surplus. This payoff is intended to capture, in a stylized way, the incentives of a large online platform that provides buyers with information about sellers’ products and earns a commission as a fraction of the trade price. At the same time, the platform may wish to leave some rents to sellers to prevent them from migrating to rival platforms.

For this section, it will be convenient to strengthen Assumption~\ref{assn:increasing_gainstop} as follows:
\begin{assumption}\label{assn:contraction}
The cost function $c:\Theta \to (0,\infty)$ is strictly increasing. Moreover, for every $\beta \in [0,1]$ there exists $\theta_{\beta}$ such that $(1-\beta)c(\theta_{\beta})=\theta_{\beta}$, with $(1-\beta)c(\theta)>\theta$ for $\theta<\theta_{\beta}$ and $(1-\beta)c(\theta)<\theta$ for $\theta>\theta_{\beta}$.
\end{assumption}
Assumption~\ref{assn:contraction} strengthens Assumption~\ref{assn:increasing_gainstop} by requiring that the threshold property holds not only for the cost function $c$, but also for every adjusted cost function $(1-\beta)c$ with $\beta \in [0,1]$. That is, for each $\beta$ there exists a threshold type $\theta_{\beta}$ such that all types below $\theta_{\beta}$ are inefficient and all types above are efficient. Recall that $\beta$ captures the relative weight placed on revenue versus efficiency. When $\beta=0$, the objective reduces to efficiency and $\theta_{\beta}=\theta^*$, where $\theta^*$ is the threshold type from Assumption~\ref{assn:increasing_gainstop}. When $\beta=1$, the objective reduces to revenue and $\theta_{\beta}=0$. Moreover, since $(1-\beta)c(\theta)$ is strictly decreasing in $\beta$ for each $\theta$, the threshold $\theta_{\beta}$ is strictly decreasing in $\beta$, ranging over the entire interval $[0,\theta^*]$.

With Assumption~\ref{assn:increasing_gainstop}, we can characterize the optimal solutions to the primal problem~(\ref{primal}) under assumptions on the shape of the designer’s objective. Moreover, we exploit the arguments from Propositions~\ref{prop:alpha_increasing} and \ref{prop:alpha_convex} to obtain these characterizations. We begin with the following result:

\begin{proposition}\label{prop:revenue-surplus}
Suppose Assumption~\ref{assn:contraction} holds, and let $\underline{\theta}$ be the first type to be matched under the reveal–pool negative assortative matching distribution, $\pi^{\text{NAM}}$. Then:
\begin{enumerate}
    \item If $\theta_\beta \leq \underline{\theta}$ and the mapping $\theta \mapsto \frac{\theta-(1-\beta)c(\theta)}{x-\theta}$ is strictly increasing on $[\theta_{\beta},\theta^*]$ for every $x \in [\theta^*,c(1)]$, then $\pi^{\text{NAM}}$ is the unique solution to the primal problem~(\ref{primal}).
    
    \item If $\theta_\beta > \underline{\theta}$, a feasible joint distribution $\pi$ solves the primal problem~(\ref{primal}) if and only if it satisfies:
    \begin{enumerate}
        \item[(i)] Types $\theta < \theta_\beta$ do not trade:
        \[
        \int_0^{\theta_{\beta}}\int_0^1 \mathbf{1}_{\{x \geq \theta^*\}}\, d\pi(\theta,x) = 0.
        \]
        \item[(ii)] Types $\theta \geq \theta_\beta$ trade:
        \[
        \int_{\theta_{\beta}}^1\int_0^1 \mathbf{1}_{\{x \geq \theta^*\}}\, d\pi(\theta,x) 
        = 1 - F(\theta_{\beta}).
        \]
    \end{enumerate}
\end{enumerate}
\end{proposition}
The proof of part 2 of Proposition~\ref{prop:revenue-surplus} is provided in Appendix~\ref{sec:proof_price_surplus}. Part 1 follows directly from Proposition~\ref{prop:alpha_increasing}. Specifically, the designer’s problem is:
\[
\max_{\pi \in \Delta(\Theta \times X)} \; \int_0^1\int_0^1 [x-(1-\beta)c(\theta)] \cdot \mathbf{1}_{\{x\ge \theta^*\}} \, d\pi(\theta, x) 
\quad \text{s.t. \ref{BP}, \ref{martingale}, and \ref{price}}.
\]
Because the martingale constraint (\ref{martingale}) implies that for all measurable $B \subseteq X$:
\[
\int_0^1\int_B x\,d\pi(\theta,x) = \int_0^1\int_B \theta\,d\pi(\theta,x),
\]
we can, without loss of optimality, rewrite the objective as:
\[
\int_0^1\int_0^1 [\theta-(1-\beta)c(\theta)] \cdot \mathbf{1}_{\{x\ge \theta^*\}} \, d\pi(\theta, x).
\]

This objective coincides with the weighted volume of trade objective when $\alpha(\theta)=\theta-(1-\beta)c(\theta)$. The difference is that in Proposition \ref{prop:alpha_increasing} we assumed $\alpha(\theta)>0$ for all $\theta$, whereas here $\alpha(\theta)<0$ for $\theta<\theta_\beta$ and $\alpha(\theta)> 0$ otherwise. Since types with $\alpha(\theta)<0$ reduce the objective, it is optimal to exclude them from trade. We can therefore fully reveal these inefficient types and, conditional on that, solve the designer’s problem considering only types above $\theta_\beta$.

If $\theta_\beta \leq \underline{\theta}$ and, for every $x \in [\theta^*,c(1)]$, the mapping $\theta \mapsto \frac{\theta-(1-\beta)c(\theta)}{x-\theta}$ is strictly increasing on $[\theta_{\beta},\theta^*]$, then the argument from Proposition \ref{prop:alpha_increasing} applies, and $\pi^{\text{NAM}}$ is the unique solution. If instead $\theta_\beta > \underline{\theta}$, then $\pi^{\text{NAM}}$ is not optimal, as it allows some types with $\alpha(\theta)<0$ to trade. In this case, we can modify $\pi^{\text{NAM}}$ so that matching starts at $\theta_\beta$ rather than $\underline{\theta}$, ensuring that a type trades if and only if $\alpha(\theta)\ge 0$. However, because other joint distributions can also achieve this, the reveal-pool negative assortative matching function is not longer uniquely optimal.

Finally, by the same logic, we can extend Proposition \ref{prop:alpha_convex} to the case where the mapping $\theta \mapsto \frac{\theta-(1-\beta)c(\theta)}{x-\theta}$ is convex. The modification required is that the pool–reveal–pool negative assortative matching distribution should begin pooling at $\theta_{\beta}$ rather than at $0$. In this way, we first reveal all types $\theta \in [0,\theta_{\beta})$, for which $\theta-(1-\beta)c(\theta)<0$. 

To construct the modified distribution, we use Lemma \ref{lem:bijective-function} to guarantee the existence and uniqueness of a strictly decreasing bijection $g_{\beta}:[\theta_{\beta},\bar{\theta}_{\beta}] \to [\theta^*,c(1)]$ satisfying the differential equation (\ref{ODE}). Whenever $\theta_{\beta} \leq \underline{\theta}$, we define the modified pool–reveal–pool structure as follows. For some $x_{\beta} \in (\theta^*,c(1))$, let $\theta_{1}=g_{\beta}^{-1}(x_{\beta})$ and $\theta_2=g^{-1}_2(x)$. Denote the resulting distribution by $\pi^{x_{\beta}-\text{NAM}}$, where:
\[
    d\pi^{x_{\beta}-\text{NAM}}(\theta,x)=
    \begin{cases}
        f(\theta)\delta_{\theta}(x) \,d\theta & \text{if } \theta \in [0,\theta_{\beta}), \\
        f(\theta)\delta_{g_{\beta}(\theta)}(x) \,d\theta & \text{if } \theta \in [\theta_{\beta},\theta_1], \\
        f(\theta)\delta_{\theta}(x) \,d\theta & \text{if } \theta \in (\theta_1,\theta_2), \\
        f(\theta)\delta_{g_2(\theta)}(x) \,d\theta & \text{if } \theta \in [\theta_2,\theta^*], \\
        f(\theta)\delta_{c(\theta)}(x) \,d\theta & \text{if } \theta \in (\theta^*,1].
    \end{cases}
\]
Thus, $\pi^{x_{\beta}-\text{NAM}}$ fully reveals types below $\theta_{\beta}$, while for types above $\theta_{\beta}$ it preserves the pool–reveal–pool negative assortative structure. 

We can now state the analogue of Proposition \ref{prop:alpha_convex} when the designer’s objective is a convex combination of price and surplus:

\begin{proposition}\label{prop:revenue-surplus_convex}
Suppose Assumption~\ref{assn:contraction} holds, and let $\underline{\theta}$ be the first type to be matched under the reveal–pool negative assortative matching distribution $\pi^{\text{NAM}}$. If $\theta_\beta \leq \underline{\theta}$, the mapping $\theta \mapsto \frac{\theta-(1-\beta)c(\theta)}{x-\theta}$ is strictly convex on $[\theta_{\beta},\theta^*]$ for every $x \in [\theta^*,c(1)]$, and for $x=c(1)$ we have
\[
\frac{\theta_{\beta}-(1-\beta)c(\theta_{\beta})}{c(1)-\theta_{\beta}}
>\frac{\underline{\theta}-(1-\beta)c(\underline{\theta})}{c(1)-\underline{\theta}},
\]
then there exists $x_{\beta} \in (\theta^*,c(1))$ such that $\pi^{x_{\beta}-\text{NAM}}$ solves the primal problem~(\ref{primal}). 

Conversely, if $\pi$ is a solution to the primal problem~(\ref{primal}), then there exists $x_{\beta} \in (\theta^*,c(1))$ such that $\pi_{\theta}^{x_{\beta}-\text{NAM}}=\pi_{\theta}$ for $F$-almost every $\theta$.
\end{proposition}

\section{Multiple Intersections and Gains at the Bottom}\label{sec:extensions}
In this section we extend our main results to the case where Assumption~\ref{assn:increasing_gainstop} is violated. 
The cost function $c$ remains strictly increasing but may now intersect the $45^\circ$ line at an arbitrary finite number of points. 
To illustrate the logic of the construction, we focus on the case where the designer's objective is a weighted volume of trade, with weighting function $\alpha:\Theta\to(0,\infty)$ strictly increasing. 
When Assumption~\ref{assn:increasing_gainstop} holds, Proposition~\ref{prop:alpha_increasing} establishes that the unique optimal solution to the primal problem~(\ref{primal}) is the reveal--pool negative assortative matching distribution. 
We show that this negative assortative matching structure continues to characterize the optimal solution when there are multiple intersections, though in this case we must account for the presence of multiple intervals of efficient and inefficient types. 
Intuitively, the construction simply replicates the logic of our base case within each adjacent inefficient--efficient pair of intervals, which we make explicit through an example below.

\paragraph{Example: Three Intersections (Four Blocks)}

To illustrate the logic of the construction, consider the case where the cost function $c$ intersects the $45^\circ$ line three times at thresholds $\theta_1^* < \theta_2^* < \theta_3^*$. 
This partitions the type space $\Theta$ into four contiguous intervals that alternate between inefficient and efficient regions. 
Without loss of generality, suppose we are in the configuration:
\[
I_1 < E_2 < I_3 < E_4,
\]
where $I$ denotes an inefficient interval ($\theta<c(\theta)$) and $E$ an efficient interval ($\theta>c(\theta)$). 
We write $E_i=[\underline{\theta}_{E_i},\bar{\theta}_{E_i}]$ and analogously for inefficient intervals.

\begin{figure}[H]
    \centering
    \begin{subfigure}[t]{0.48\textwidth}
        \centering
        \includegraphics[width=\linewidth]{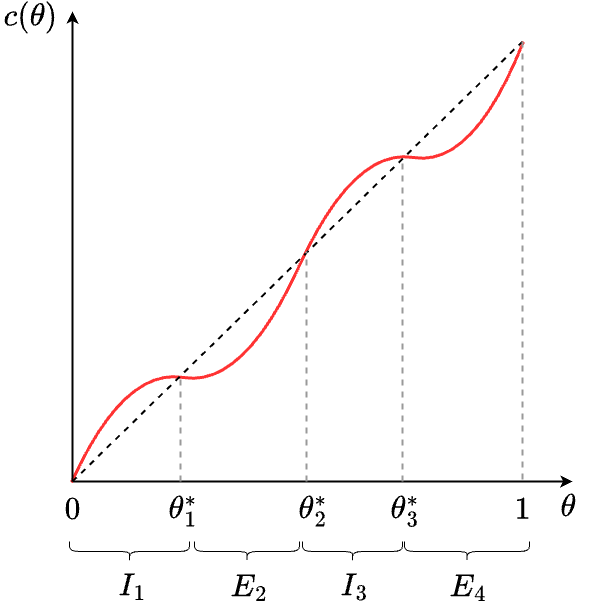}
        \caption{Cost function.}
        \label{fig:cost_multiple_cross}
    \end{subfigure}
    \hfill
    \begin{subfigure}[t]{0.48\textwidth}
        \centering
        \includegraphics[width=\linewidth]{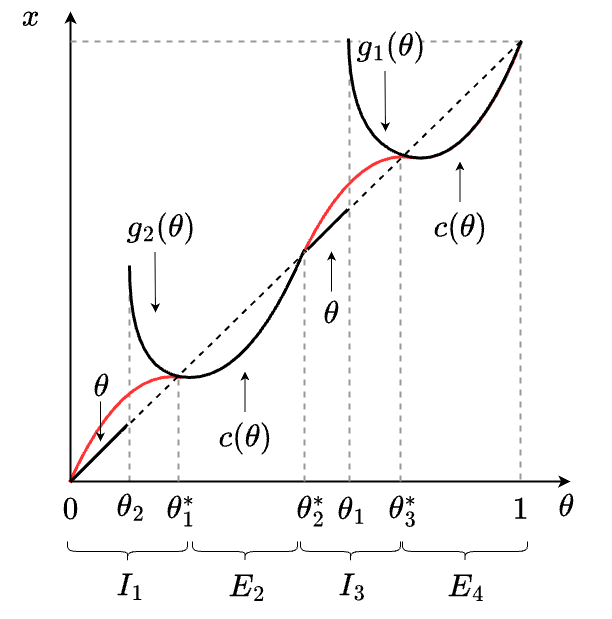}
        \caption{Cost function (in red) and optimal signal (in black).}
        \label{fig:signal_multiple_cross}
    \end{subfigure}
    \caption{Optimal signal for increasing $\alpha(\theta)$ when the cost function intersects the 45-degree line multiple times}
    \label{fig:multiple_cross}
\end{figure}

\paragraph{Step 1. Matching $I_3$ with $E_4$.}
We first replicate the logic of the case where Assumption \ref{assn:increasing_gainstop} holds on the highest adjacent pair $I_3<E_4$. 
We can use Lemma \ref{lem:bijective-function} to define the unique strictly decreasing bijection: $g_1:[\theta_1,\theta_3^*]\to[\theta_3^*,c(\bar{\theta}_{E_4})]$, and let $\theta'_1=\max\{\theta_1,\underline{\theta}_{I_3}\}$ to ensure the domain remains inside $I_3$. 
Then:
\begin{itemize}
    \item All types $\theta \in [\underline{\theta}_{I_3},\theta'_1)$ are \emph{revealed}, i.e.\ matched with $x=\theta$.
    \item All types $\theta \in [\theta'_1,\theta_3^*]$ are \emph{pooled} with efficient types via $x=g_1(\theta)$.
    \item All types $\theta \in [\theta_3^*,c^{-1}(g_1^{-1}(\theta'_1))]\subseteq E_4$ are matched with $x=c(\theta)$.
\end{itemize}
This step may leave a residual subset of efficient types in $E_4$, namely $[c^{-1}(g_1^{-1}(\theta'_1)),\bar{\theta}_{E_4}]$.

\paragraph{Step 2. Matching $I_1$ with $E_2$.}
Next we turn to the left pair $I_1<E_2$. Two possibilities arise:
\begin{enumerate}
    \item If $\theta_1<\underline{\theta}_{I_3}$, then $E_4$ left a residual. We use Lemma \ref{lem:bijective-function} again to construct another bijection $g_2:[\theta_2,\theta_2^*]\to[g_1(\theta'_1),c(\bar{\theta}_{E_4})]$,
    and pool the tail of $I_1$ with this residual of $E_4$. Assumption \ref{assn:no_full_trade} which states that no full trade is possible implies that $I_1$ cannot be entirely covered this way. We then proceed to pair the remainder of $I_1$ with $E_2$ via $g_3:[\theta_3,\theta_2]\to[\theta_2^*,c(\bar{\theta}_{E_2})]$.
    \item If $\theta_1\geq\underline{\theta}_{I_3}$, then there is no residual from $E_4$, and we immediately use the strictly decreasing bijection $g_3:[\theta_3,\theta_2^*]\to[\theta_2^*,c(\bar{\theta}_{E_2})]$,
    to match $I_1$ with $E_2$.
\end{enumerate}

\paragraph{Conclusion.} 
In either case, the construction terminates once all efficient intervals have been exhausted.
At each step, we simply apply the case where Assumption~\ref{assn:increasing_gainstop} holds to the rightmost available adjacent pair—revealing the tail of the inefficient block if necessary and pooling the remainder with the efficient block.

\subsection{General Greedy Construction Procedure}

We now formalize the construction for the case where $c$ intersects the $45^\circ$ line at finitely many points.

Let $0<\theta_1^*<\dots<\theta_m^*<1$ denote the intersection points of $c$ with the $45^\circ$ line. 
These partition $\Theta$ into $m+1$ contiguous intervals $B_1<\dots<B_{m+1}$, alternating between inefficient ($\theta<c(\theta)$) and efficient ($\theta>c(\theta)$) regions. 
Write each interval as $B_k=[\underline{\theta}_{B_k},\bar{\theta}_{B_k}]$, with $B \in \{I,E\}$, where $I$ denotes an inefficient interval and $E$ an efficient interval.

\begin{center}
  \captionof{algorithm}{Greedy Construction}\label{greedy}
\begin{algorithmic}[1]
\begin{enumerate}
    \item \textbf{Initialization.} Let $\mathbf{E}_1$ be the rightmost efficient block (i.e., $\mathbf{E}_1=E_{m+1}$ if $B_{m+1}=E$, else $\mathbf{E}_1=E_m$). Let $\mathbf{I}_1$ be the inefficient block immediately to the \emph{left} of $\mathbf{E}_1$ (i.e., $\mathbf{I}_1=I_m$ if $\mathbf{E}_1=E_{m+1}$, else $\mathbf{I}_1=I_{m-1}$). 
    \item Process the pair $(\mathbf{I}_k,\mathbf{E}_k)$, where $\mathbf{I}_k=[\underline{\theta}_{\mathbf{I}_k},\bar{\theta}_{\mathbf{I}_k}]$ and $\mathbf{E}_k=[\underline{\theta}_{\mathbf{E}_k},\bar{\theta}_{\mathbf{E}_k}]$:
    \begin{enumerate}
        \item Use Lemma \ref{lem:bijective-function} to construct the strictly decreasing bijection:
        \[
            g_k:[\theta_k,\bar{\theta}_{\mathbf{I}_k}] \to[\underline{\theta}_{\mathbf{E}_k},\,c(\bar{\theta}_{\mathbf{E}_k})],\qquad g_k(\bar{\theta}_{\mathbf{I}_k})=\underline{\theta}_{\mathbf{E}_k}.
        \]
        \item If $\theta_k \geq \underline{\theta}_{\mathbf{I}_k}$:
        \begin{enumerate}
            \item Reveal all types in $[\underline{\theta}_{\mathbf{I}_k},\theta_k)$ by assigning $x=\theta$.
            \item Pool the types in $[\theta_k,\bar{\theta}_{\mathbf{I}_k}]$ with the types in $[\underline{\theta}_{\mathbf{E}_k},\bar{\theta}_\mathbf{E_k}]$ using $g_k$, assigning:
            \[
            x=g_k(\theta)\quad\text{for }\theta\in[\theta_k,\bar{\theta}_{\mathbf{I}_k}],\qquad x=c(\theta)\quad\text{for }\theta\in[\underline{\theta}_{\mathbf{E}_k},\bar{\theta}_{\mathbf{E}_k}].
            \]
            \item Set $\mathbf{E}_{k+1}$ equal to the next inefficient interval to the left of $\mathbf{E}_{k}$. That is, $\mathbf{E}_{k+1}=E_j$, where $E_j<\mathbf{E}_{k}$ and $E_j>E_i$ for any other $E_i<\mathbf{E}_{k}$.
            \item If $\mathbf{E}_{k+1} > \mathbf{I}_k$ and $\theta_k>\underline{\theta}_{\mathbf{I}_k}$:
            \begin{enumerate}
                \item Set $\mathbf{I}_{k+1}=[\underline{\theta}_{\mathbf{I}_k},\theta_k]$.
            \end{enumerate}
            \item If $\mathbf{E}_{k+1} < \mathbf{I}_k$ or $\theta_k=\underline{\theta}_{\mathbf{I}_k}$:
            \begin{enumerate}
                \item Set $\mathbf{I}_{k+1}$ equal to the next inefficient interval to the left of $\mathbf{I}_{k}$. That is, $\mathbf{I}_{k+1}=I_j$, where $I_j<\mathbf{I}_{k}$ and $I_j>I_i$ for any other $I_i<\mathbf{I}_{k}$.
            \end{enumerate}
        \end{enumerate}
        \item If $\theta_k < \underline{\theta}_{\mathbf{I}_k}$:
        \begin{enumerate}
            \item Pool the types in $[\underline{\theta}_{\mathbf{I}_k},\bar{\theta}_{\mathbf{I}_k}]$ with the types in $[\underline{\theta}_{\mathbf{E}_k},c^{-1}(g(\underline{\theta}_\mathbf{I_k}))]$ using $g_k$, assigning:
            \[
            x=g_k(\theta)\quad\text{for }\theta\in[\underline{\theta}_{\mathbf{I}_k},\bar{\theta}_{\mathbf{I}_k}],\qquad x=c(\theta)\quad\text{for }\theta\in[\underline{\theta}_{\mathbf{E}_k},c^{-1}(g(\underline{\theta}_\mathbf{I_k}))].
            \]
            \item Set $\mathbf{E}_{k+1}$ equal to the residual types not used in the pooling of $\mathbf{E}_{k}$. That is, $\mathbf{E}_{k+1}=[c^{-1}(g(\underline{\theta}_\mathbf{I_k})),\bar{\theta}_{\mathbf{E}_k}]$.
            \item Set $\mathbf{I}_{k+1}$ equal to the next inefficient interval to the left of $\mathbf{I}_{k}$. That is, $\mathbf{I}_{k+1}=I_j$, where $I_j<\mathbf{I}_{k}$ and $I_j>I_i$ for any other $I_i<\mathbf{I}_{k}$.
        \end{enumerate}
    \end{enumerate}
    \item Repeat until there are no more efficient intervals $E_j$. All types in the remaining inefficient intervals are revealed by setting $x=\theta$.
\end{enumerate}
\end{algorithmic}
\end{center}

\begin{proposition}\label{prop:extension_multiple_crossings}
    Suppose that $c:\Theta \to (0,\infty)$ intersects the $45^\circ$ line at finitely many points, and that for every $x \in [\theta^*,c(1)]$ the mapping $\theta \mapsto \frac{\alpha(\theta)}{x-\theta}$
    is strictly increasing on $[0,\theta^*]$. 
    Then Algorithm~\ref{greedy} produces a joint distribution $\pi$ that is the uniquely optimal solution to the primal problem~(\ref{primal}).
\end{proposition}
\begin{proof}[Proof sketch]
The optimality of the construction follows by iterating the argument used in the proof of Proposition~\ref{prop:alpha_convex}.
More precisely, for each efficient interval $E_j$ and all inefficient intervals to its left, one can apply the same reasoning as in Proposition~\ref{prop:alpha_convex} to establish that the reveal–pool negative assortative matching structure is uniquely optimal.
Repeating this argument across all efficient intervals shows that the greedy construction yields the unique optimal solution for any finite number of intersections
\end{proof}

Moreover, note that throughout this section we have been agnostic about both the number of intersections and the ordering of efficient and inefficient blocks. 
The construction therefore also covers the case where there is a single intersection $\theta^*$ such that all types below $\theta^*$ are efficient and all types above $\theta^*$ are inefficient. 
In other words, rather than having \emph{gains at the top} as in Assumption~\ref{assn:increasing_gainstop}, we now have \emph{gains at the bottom}. 
In this case, there are no adjacent inefficient--efficient blocks of the form $I_{k}<E_{k+1}$, and hence the optimal joint distribution simply reveals every type. 

Intuitively, the martingale  condition~(\ref{martingale}) and the prices-as-means condition~(\ref{price}) prevent any inefficient type $\theta$ from trading. 
Pooling $\theta$ with an efficient type $\theta'$ always produces a posterior mean $x<\theta$. 
As $c(\theta)>\theta>x$, $\theta$ would not accept to trade at price $x$. 
Thus, unlike in the gains at the top case, the designer cannot exploit efficient types to enable trade for inefficient ones. 
This reasoning applies regardless of the designer’s objective. 
We therefore obtain the following corollary.

\begin{corollary}\label{coro:extension_gainsatthebottom}
    Suppose the cost function $c$ is strictly increasing and there exists $\theta^*$ such that: 
    \[
    c(\theta^*)=\theta^*,\quad c(\theta)<\theta \;\;\text{for all } \theta<\theta^*,\quad c(\theta)>\theta \;\;\text{for all } \theta>\theta^*.
    \] 
    Then, for any designer's objective, the fully revealing joint distribution $\pi^{\text{id}}$ solves the primal problem \ref{primal}, where:
    \[
    d\pi^{\text{id}}(\theta,x)=f(\theta)\,\delta_{\theta}(x)\,d\theta.
    \]
\end{corollary}

\section{Conclusion}\label{sec:conclusion}

We have studied how a policymaker or intermediary can use information disclosure to manipulate market outcomes in competitive markets with adverse selection. Motivated by settings where the objective extends beyond ensuring the efficiency of each individual trade—such as financial stability policies, initiatives aimed at increasing financial inclusion, and intermediated markets—we have characterized the disclosure policies that maximize volume of trade and a combination of price and producer surplus.

The presence of adverse selection implies that even after the designer's signal is observed by market participants, sellers' decisions to trade reveal additional information about the quality of their products. Thus, the designer faces endogenous constraints in the set of posterior beliefs she can induce. We address this by reformulating the problem as a martingale optimal transport exercise with a free marginal, and subject to an additional constraint that ensures markets do not infer additional information from sellers' willingness to trade. Complementary slackness under strong duality provides a sharp characterization of the support of the primal solution, allowing us to identify the set of candidate optimal signals. We then explicitly construct dual multipliers to support these signals as solutions. 

We show that, when the designer maximizes weighted volume of trade, the optimal signals feature negative assortative matching between inefficient and efficient sellers, taking either a reveal–pool or a pool–reveal–pool form depending on whether the objective places larger weight on higher or lower types. This signal remains uniquely optimal when the designer maximizes a weighted sum of transaction price and surplus, provided the weight on price is high enough. As the designer cares more about efficiency, any signal that fully discloses low‐quality types while pooling all middle types with some high‐quality types becomes optimal. Finally, we show that when the cost structure is such that there are multiple efficient and inefficient pools of sellers, the analysis applies within each inefficient-efficient pair of regions: the optimal policy repeats the reveal–then–pool construction, yielding separate negative-assortative pools.

\newpage
\bibliography{references.bib}
\bibliographystyle{aer}

\appendix

\newpage

\section{Preliminary Results}\label{sec:preliminary}

\begin{lemma}[Existence and Uniqueness of Decreasing Bijections]\label{lem:bijective-function}
   There exist unique functions  $g_1 : [0, \bar{\theta}] \to [\theta^*, c(1)]$ and $g_2 : [\underline{\theta}, \theta^*] \to [\theta^*, c(1)]$ that are strictly decreasing, differentiable, and bijective, where $0<\bar{\theta}<\underline{\theta}<\theta^*$, such that:
    \begin{equation}
        x=g_i^{-1}(x)\frac{-f\left(g_i^{-1}(x)\right)\frac{\partial g^{-1}(x)}{\partial x}}{-f\left(g_i^{-1}(x)\right)\frac{\partial g_i^{-1}(x)}{\partial x}+f\left(c^{-1}(x)\right)\frac{\partial c^{-1}(x)}{\partial x}}+c^{-1}(x)\frac{f\left(c^{-1}(x)\right)\frac{\partial c^{-1}(x)}{\partial x}}{-f\left(g_i^{-1}(x)\right)\frac{\partial g_i^{-1}(x)}{\partial x}+f\left(c^{-1}(x)\right)\frac{\partial c^{-1}(x)}{\partial x}},
    \end{equation}
    for $i=1,2$ and all $x \in [\theta^*,c(1)]$.
\end{lemma}

\begin{proof}[Proof of Lemma~\ref{lem:bijective-function}]
   We show the existence and uniqueness of the function $g_2:[\underline{\theta},\theta^*] \to [\theta^*,c(1)]$, where $\underline{\theta} \in (0,\theta^*)$. The proof for the function $g_1$ is omitted, since it follows essentially the same argument. Note also that, since they solve the same terminal value problem, $g$ and $g_2$ are identical, except for the initial point in the domain, which is $\underline{\theta}$ for $g$ and $\theta_2$ for $g_2$.
    
   First, we extend the density function $f$ so that it is defined over the interval $(-\infty, 1]$, remains strictly positive everywhere, and satisfies $f \in C^{1}((-\infty, 1])$. Similarly, we extend the cost function $c$ so that it is defined over $[0, \infty)$, is strictly increasing, and satisfies $c \in C^{1}([0, \infty))$.

    Next, let $b(x):= c^{-1}(x)$. Since $c$ is strictly increasing and continuously differentiable, $b$ is also strictly increasing and continuously differentiable. Define the domain $D = [\theta^*, \infty) \times (-\infty, \theta^*]$, and let $G : D \to \mathbb{R}$ be the function defined by:
    \[
        G(x,a)=
        \begin{cases}
            \frac{b'(x)f(b(x))}{f(a)}\frac{b(x)-x}{a-x} & \text{if } (x,a) \in D \setminus \{(\theta^*,\theta^*)\}, \\
            0 & \text{if } (x,a)=(\theta^*,\theta^*).\
        \end{cases}
    \]
    
    The function $G$ is continuous because it is composed of continuous functions $b$ and $f$, and it satisfies $\lim_{(x,a) \to (\theta^*, \theta^*)} G(x,a) = 0$. Indeed, since $c(\theta) > \theta$ for all $\theta \in [0, \theta^*)$, $c(\theta^*) = \theta^*$, and $c(\theta) < \theta$ for all $\theta \in (\theta^*, 1]$, it follows that $c'(\theta) < 1$ in a neighborhood around $\theta^*$. This implies that $b'(\theta) > 1$ in a neighborhood around $\theta^*$, and therefore $b(x) - x$ converges to zero faster than $a - x$ as $(x,a) \to (\theta^*, \theta^*)$. Moreover, since both $b$ and $f$ are continuously differentiable, we conclude that $F$ is locally Lipschitz.

    Define the initial value problem given by:
    \begin{equation}
        a'(x)=G(x,a(x)), \quad a(\theta^*)=\theta^*. \tag{IVP} \label{eq:IVP}
    \end{equation}

    Since $G$ is continuous and locally Lipschitz, we can apply the Picard--Lindelöf Theorem to conclude that there exists a unique solution to~\ref{eq:IVP}, given by $a : [\theta^*, \theta^* + \varepsilon] \to (-\infty, \theta^*]$ for some $\varepsilon > 0$, where $a(x) < \theta^*$ and, since $b(x)>x$, $b'(x)>0$ and $f$ is strictly positive, we have $a'(x)=F(x,a(x))<0$ for all $x \in (\theta^*, \theta^* + \varepsilon]$.

    Since the solution remains in the domain $D$ and it is strictly decreasing, we can use a continuation argument to extend the unique solution to the entire interval $[\theta^*, c(1)]$, where $a(x) < \theta^* < x$ and $a'(x) = F(x, a(x)) < 0$ for all $x \in (\theta^*, c(1)]$.

   Let $\underline{\theta} = a^{-1}(c(1))$ and define $g_2(x) := a^{-1}(x)$. Then $g_2 : [\underline{\theta}, \theta^*] \to [\theta^*, c(1)]$ is the unique strictly decreasing, differentiable, and bijective function that satisfies:
    \[
        \frac{\partial g_2^{-1}(x)}{\partial x} = G(x, g_2^{-1}(x)), \quad \text{for all } x \in [\theta^*, c(1)].
    \]
    
    It remains to show that $\underline{\theta} \in (0, \theta^*)$. Suppose, for the sake of contradiction, that $\underline{\theta} \leq 0$. Then define the following $\pi \in \Delta(\Theta \times X)$:
     \[
        d\pi(\theta,x)=
        \begin{cases}
            f(\theta)\delta_{g_2(\theta)}(x) \,d\theta & \text{if } \theta \in [0,\theta^*], \\
            f(\theta)\delta_{c(\theta)}(x) \,d\theta & \text{if } \theta \in \left[\theta^*,c^{-1}(g(0))\right], \\
            f(\theta)\delta_{\theta}(x) \,d\theta & \text{if } \theta \in \left(c^{-1}(g(0)),1\right].
        \end{cases}
    \]    
    In essence, $\pi$ matches each type $\theta \in [0, \theta^*]$ with the posterior mean $g_2(\theta)$, each type $\theta \in \left[\theta^*, c^{-1}(g(0))\right]$ with the posterior mean $c(\theta)$, and each type $\theta \in \left(c^{-1}(g_2(0)), 1\right]$ with a posterior mean equal to its type.

   Since $\underline{\theta} \leq 0$, we have $g_2(0) \leq c(1)$. Moreover, it is easy to verify that $\pi$ satisfies conditions~(\ref{BP}),~(\ref{martingale}), and~(\ref{price}) (see the proof of Lemma~\ref{lem:nam-feasible} for details). This implies that $\pi$ is a feasible joint distribution in which all types trade, as each is matched with a posterior mean $x \geq \theta^*$, contradicting Assumption~\ref{assn:no_full_trade}.

   An equivalent argument shows the existence and uniqueness of $g_1:[0,\bar{\theta}] \to [\theta^*,c(1)]$ with $\bar{\theta} \in (0,\theta^*)$. We provide a brief argument showing that $\bar{\theta} < \underline{\theta}$. Suppose instead that $\bar{\theta} \geq \underline{\theta}$. By the intermediate value theorem, there exists $\hat{\theta} \in [\underline{\theta},\bar{\theta}]$ such that $g_1(\hat{\theta})=g_2(\hat{\theta})$. We can then construct a distribution $\pi \in \Delta(\Theta \times X)$ that matches each type $\theta \in [0,\hat{\theta}]$ with posterior mean $g_1(\theta)$, each type $\theta \in (\hat{\theta},\theta^*]$ with posterior mean $g_2(\theta)$, and each type $\theta \in (\theta^*,1]$ with posterior mean $c(\theta)$. The resulting $\pi$ is feasible (see the proof of Lemma~\ref{lem:nam-feasible}), and as in the previous case, all types get to trade---contradicting Assumption~\ref{assn:no_full_trade}.
\end{proof}

\begin{proof}[Proof of Lemma~\ref{lem:nam-feasible}]
    We show that $\pi^{\text{NAM}}$ satisfies conditions~(\ref{BP}),~(\ref{martingale}), and~(\ref{price}). The proof for $\pi^{x^*-\text{NAM}}$ is omitted, since it follows essentially the same argument.
    
    For every measurable $A \subseteq \Theta$ it is immediate that:
    \[
        \int_A \int_X d\pi^{\text{NAM}}(\theta,x)=\int_A f(\theta) \, d\theta.
    \]
    
    Therefore, $\pi^{\text{NAM}}$ satisfies condition~(\ref{BP}). Moreover, since no type $\theta$ is matched with a posterior mean $x < \hat{c}(\theta)$, it follows that for every measurable set $B \subseteq X$:
    \[
        \int_0^1 \int_B \mathbf{1}_{\{\hat{c}(\theta)>x\}}d\pi^{\text{NAM}}(\theta,x)=0.
    \]

    Thus, $\pi^{\text{NAM}}$ also satisfies condition~(\ref{price}).

    See that for every $[x_1,x_2] \subseteq [0,\underline{\theta})$ we have that:
    \[
        \int_0^1\int_{x_1}^{x_2}(x-\theta)\,d\pi^{\text{NAM}}(\theta,x)=\int_{x_1}^{x_2}(\theta-\theta)f(\theta)\,d\theta=0.
    \]

    Additionally, for every $[x_1,x_2] \subseteq [\underline{\theta},c(1)]$ it follows that:
    \begin{align*}
        \int_0^1\int_{x_1}^{x_2}(x-\theta)\,d\pi^{\text{NAM}}(\theta,x)&=\int_{\underline{\theta}}^{\theta^*}\int_{x_1}^{x_2}(x-\theta)f(\theta)\delta_{g_2(\theta)}(x)\,d\theta+\int_{\theta^*}^{1}\int_{x_1}^{x_2}(x-\theta)f(\theta)\delta_{c(\theta)}(x)\,d\theta \\
        &=\int_{g_2^{-1}(x_1)}^{g_2^{-1}(x_2)}(g_2(\theta)-\theta)f(\theta)\,d\theta+\int_{c^{-1}(x_1)}^{c^{-1}(x_2)}(c(\theta)-\theta)f(\theta)\,d\theta \\
        &=\int_{x_1}^{x_2}(g_2^{-1}(x)-x)f(g_2^{-1}(\theta))\frac{\partial g_2^{-1}(x)}{\partial x}\,dx+\int_{x_1}^{x_2}(x-c^{-1}(x))f(c^{-1}(\theta))\frac{\partial c^{-1}(x)}{\partial x}\,dx \\
        &=\int_{x_1}^{x_2}\left[\frac{\partial g_2^{-1}(x)}{\partial x}-\frac{\frac{\partial c^{-1}(x)}{\partial x}f(c^{-1}(\theta))}{f(g_2^{-1}(\theta))}\frac{c^{-1}(x)-x}{g_2^{-1}(x)-x}\right](g_2^{-1}(x)-x)f(g_2^{-1}(x))\,dx \\
        &=\int_{x_1}^{x_2}\left[\frac{\partial g_2^{-1}(x)}{\partial x}-G(x,g_2^{-1}(x))\right](g_2^{-1}(x)-x)f(g_2^{-1}(x))\,dx \\
        &=0,
    \end{align*}
    where the third equality follows from the change of variables $x = g_2(\theta)$ and $x = c(\theta)$, and the last equality follows from the fact that $g_2$ satisfies the differential equation (Lemma~\ref{lem:bijective-function}):
    \[
    \frac{\partial g_2^{-1}(x)}{\partial x}=G(x,g_2^{-1}(x)).
    \]

    Finally, since no $x \in [\underline{\theta},\theta^*) \cup (c(1),1]$ lies in the support of $\pi^{\text{NAM}}$, we can conclude that $\pi^{\text{NAM}}$ also satisfies condition~(\ref{martingale}).
\end{proof}

\section{Proofs of Section~\ref{sec:volume_trade}}\label{sec:proof_volume}

\paragraph{Outline of the Proofs.} To characterize optimal solutions to the primal problem~(\ref{primal}) (Propositions~\ref{prop:alpha_increasing} and~\ref{prop:alpha_convex}) we proceed in two steps:

\begin{enumerate}[label=\textbf{Step \arabic*:}, leftmargin=*]

    \item Assume, for this step only, that strong duality holds: there exist a solution $\pi^*$ to the primal~\ref{primal} and a solution $(w^*,q^*,m^*)$ to the dual~\ref{dual}, with zero duality gap. Using complementary slackness (Lemma~\ref{cor:complementary-slackness}), we prove a sequence of lemmas that pin down the structure of $\pi^*$. Specifically, we show that either $\pi^*_{\theta}=\pi_{\theta}^{\text{NAM}}$ for $F$-almost every $\theta\in\Theta$ (Propositions~\ref{prop:alpha_increasing}), or there exists $x^*\in[\theta^*,c(1)]$ such that $\pi^*_{\theta}=\pi_{\theta}^{x^*-\text{NAM}}$ for $F$-almost every $\theta\in\Theta$ (Proposition~\ref{prop:alpha_convex}).

    \item We explicitly construct functions $(w,q,m)$ that are feasible for the dual~\ref{dual} (i.e., they satisfy~\ref{ZP}) and for which the dual objective equals the primal objective evaluated at $\pi^{\text{NAM}}$ (Propositions~\ref{prop:alpha_increasing}) or at $\pi^{x^*-\text{NAM}}$ (Proposition~\ref{prop:alpha_convex}). By Lemma~\ref{cor:strong-duality}, this certifies optimality of the corresponding primal allocation and establishes strong duality, thereby validating the assumption used in Step~1.
\end{enumerate}

\noindent This ordering is intentional: the structural insight from Step~1 guides the construction of the dual multipliers in Step~2.

\paragraph{Notation.}
Throughout Step~1 we use the following notation. For each type $\theta$, let $x^*(\theta)$ denote the set of posterior means with which $\theta$ is matched under the optimal joint distribution $\pi^*$; equivalently,
$x^*(\theta)=\operatorname{supp}\,(\pi_{\theta})$. We say that $\theta$ is \emph{revealed} if $\theta \in x^*(\theta)$, and \emph{fully revealed} if $x^*(\theta)=\{\theta\}$. Conversely, we say that $\theta$ is \emph{pooled} with some $\theta'\neq\theta$ if there exists some $x \notin \{\theta,\theta'\}$ such that $x \in x^*(\theta)\cap x^*(\theta')$.

\subsection{Proof of Proposition \ref{prop:alpha_increasing}}\label{sec:proof_alpha_increasing}

\paragraph{Step 1.} Recall that for this step, we assume that strong duality holds. Therefore, let $\pi^*$ be any solution to the primal problem~(\ref{primal}), and let $(w^*, q^*, m^*)$ be any solution to the dual problem~(\ref{dual}). 

Moreover, recall that our \textbf{complementary slackness} result (Lemma~\ref{cor:complementary-slackness}), implies that for $F$-almost every $\theta \in \Theta$:
\[
    w^*(\theta) = \sup_{x \in X} \left\{ \alpha(\theta)\,\mathbf{1}_{\{x \geq \theta^*\}} + q^*(x)(x - \theta) + m^*(x)\,\mathbf{1}_{\{\hat{c}(\theta) > x\}} \right\},
\]
and for $\pi^*$-almost every $(\theta, x) \in \Theta \times X$:
\[
    w^*(\theta) = \alpha(\theta)\,\mathbf{1}_{\{x \geq \theta^*\}} + q^*(x)(x - \theta) + m^*(x)\,\mathbf{1}_{\{\hat{c}(\theta) > x\}}.
\]

We will leverage the complementary slackness result above to show that $\pi^*_{\theta}=\pi^{\text{NAM}}_{\theta}$ for $F$-almost every $\theta \in \Theta$.

\begin{lemma}\label{lem:q_negative}
The function $q^*$ satisfies $q^*(x) < 0$ for all $x \in [\theta^*,1]$.
\end{lemma}

\begin{proof}
    By Assumption \ref{assn:no_full_trade} a subset of types in $[0, \theta^*)$ is revealed, and this subset has positive measure under $F$. Then, by complementary slackness, for any $x \geq \theta^*$ and for almost every $\theta$ in that subset, we must have:
    \[
        0 \geq \alpha(\theta) + q^*(x)(x - \theta) \quad \Rightarrow \quad q^*(x) \leq -\frac{\alpha(\theta)}{x - \theta}<0,
    \]
    where the strict inequality follows from the fact that $\alpha$ is strictly positive.
\end{proof}

\begin{lemma}\label{lem:alpha_pool_high}
    Under the optimal joint distribution $\pi^*$, $F$-almost every $\theta \in (\theta^*, 1]$ is pooled (with probability one) with some type $\theta' \in [0, \theta^*)$ to achieve a posterior mean equal to its cost, $x^*(\theta) = c(\theta)$.
\end{lemma}
\begin{proof}
    Suppose, toward a contradiction, that a subset of types in $(\theta^*, 1]$ is revealed, and that this subset has positive $F$-measure. Then, by complementary slackness, for almost every $\theta$ in that subset, the value of revealing $\theta$ must be at least as high as the value of being matched to a posterior mean $x \geq c(\theta)$:
    \[
        \alpha(\theta) \geq \alpha(\theta) + q^*(x)(x - \theta).
    \]
    In particular, for $x \in [c(\theta), \theta)$, the inequality implies $q^*(x) \geq 0$, contradicting Lemma~\ref{lem:q_negative}.

   Next, suppose a type $\theta \in (\theta^*, 1]$ is pooled with another type $\theta' \in [\theta^*, 1]$. By the martingale condition~(\ref{martingale}), the resulting posterior mean $x$ must, for $F$-almost every such pooling, be strictly greater than at least one of the two types. Without loss of generality, assume $x > \theta$. For type $\theta$ to prefer this pooling over being revealed, it must be that:
    \[
        \alpha(\theta) + q^*(x)(x - \theta) \geq \alpha(\theta),
    \]
    which again implies $q^*(x) \geq 0$, contradicting Lemma~\ref{lem:q_negative}. Therefore, $F$-almost every $\theta \in (\theta^*, 1]$ must be pooled with some type $\theta' < \theta^*$. 

    Finally, suppose a type $\theta \in (\theta^*,1]$ is pooled with a type $\theta' \in [0,\theta^*)$ to form a posterior mean $x \neq c(\theta)$. By the prices-as-means constraint~(\ref{price}) and the martingale condition~(\ref{martingale}), for $F$-almost every such pooling, it must be that $x \in (c(\theta),\theta)$. Now, consider any $x' \in [c(\theta), x)$. By complementary slackness, we must have:
    \begin{align*}
        \alpha(\theta) + q^*(x)(x - \theta) &\geq \alpha(\theta) + q^*(x')(x' - \theta), \\
        \alpha(\theta') + q^*(x)(x - \theta') &\geq \alpha(\theta') + q^*(x')(x' - \theta').
    \end{align*}
    Since $0 > x-\theta > x'-\theta$ and $q(x)$, $q(x') < 0$, the first inequality implies $q(x) < q(x')$. On the other hand, because $x-\theta' > x'-\theta' > 0$ and $q(x)$, $q(x') < 0$, the second equality implies $q(x) > q(x')$, a contradiction.  
    Therefore, for $F$-almost every $\theta \in (\theta^*, 1]$, the posterior mean must satisfy $x = c(\theta)$.
\end{proof}

\begin{lemma}\label{lem:alpha_pooling_interval}
    There exists a type $\uwave{\theta} \in [0, \theta^*)$ such that under the optimal joint distribution $\pi^*$, $F$-almost every $\theta \in [0, \uwave{\theta})$ is fully revealed, and $F$-almost every $\theta \in [\uwave{\theta}, \theta^*)$ is pooled (with probability one).
\end{lemma}
\begin{proof}
    Let $A$ denote the set of types in $[0, \theta^*)$ that are pooled. By Lemma~\ref{lem:alpha_pool_high}, $F(A)>0$. Suppose $A$ cannot be written as an interval of the form $[\uwave{\theta}, \theta^*)$ for some $\uwave{\theta} >0$. Then, we can find disjoint subsets $A_1, A_2 \subset [0, \theta^*)$, each with positive measure, such that: all types in $A_1$ are pooled, all types in $A_2$ are fully revealed, and for every $\theta_1 \in A_1$ and $\theta_2 \in A_2$, we have $\theta_2 > \theta_1$.

    Let $x \geq \theta^*$ denote the posterior mean assigned to some $\theta_1 \in A_1$. By complementary slackness, we have for $F$-almost every $\theta_1 \in A_1$:
    \[
       \alpha(\theta_1) + q^*(x)(x - \theta_1) \geq 0 \quad \Rightarrow \quad \frac{\alpha(\theta_1)}{x-\theta_1} \geq -q^*(x).
    \]
    
    Since, by assumption, $\theta \mapsto \frac{\alpha(\theta)}{x-\theta}$ is strictly increasing on the interval $[0,\theta^*]$. It follows that for almost every $\theta_2 \in A_2$:
    \[
        \frac{\alpha(\theta_2)}{x-\theta_2} > -q^*(x) \quad \Rightarrow \quad \alpha(\theta_2) + q^*(x)(x - \theta_2) >  0,
    \]
    which contradicts complementary slackness, since $\theta_2$ is revealed and the value of being matched with $x \geq \theta^*$ exceeds the value of being revealed. Therefore, $A = [\uwave{\theta}, \theta^*)$ and $F$-almost every type $\theta \in A$ is pooled with probability one.
\end{proof}

\begin{lemma}\label{lem:alpha_singleton_decreasing}
     Under the optimal joint distribution $\pi^*$, $x^*(\theta)$ is a singleton and strictly decreasing for $F$-almost every $\theta \in [\uwave{\theta}, \theta^*)$.
\end{lemma}
\begin{proof}
    We begin by showing that $q^*$ must be strictly increasing for every $x \in (\theta^*,c(1)]$. Consider a type $\theta\in [\uwave{\theta},\theta^*)$ matched with some $x'>\theta^*$. In order for the value of matching $\theta$ with $x'$ to be greater than the value of matching $\theta$ with any $x \in [\theta^*,x')$ we must have:
    \[
        \alpha(\theta)+q^*(x')(x'-\theta)\geq \alpha(\theta)+q^*(x)(x-\theta).
    \]
    Since $x'-\theta>x-\theta>0$ and $q^*(x')$, $q^*(x)<0$, to satisfy the above inequality we need $q^*(x')>q^*(x)$. Therefore, we must have $q^*$ strictly increasing for every $x \in (\theta^*,c(1)]$. 

    Moreover, by complementary slackness, for $F$-almost every $\theta \in [\uwave{\theta},\theta^*)$ we have that:
    \[
        x^*(\theta) \subseteq \arg\max_{x \in [\theta^*,1]} \; \{\alpha(\theta) + q^*(x)(x - \theta)\}.
    \]
   
   The fact that $q^*$ is strictly increasing implies that $\alpha(\theta) + q^*(x)(x - \theta)$ is strictly sub-modular in $(\theta,x)$. By Topkis' monotone comparative statics theorem we can then conclude that $x^*$ is decreasing in the strong set order for almost every $\theta \in [\uwave{\theta},\theta^*)$.

   To finalize the proof, we must rule out two possible cases:

   \medskip

    \noindent\textbf{Case 1: $x^*(\theta)$ is not a singleton for a positive $F$-measure subset of $[\uwave{\theta}, \theta^*)$.}  
    If this is the case, then by Lemma~\ref{lem:alpha_pool_high} and the fact that $x^*$ is strictly decreasing in the strong set order, it follows that for almost every such $\theta$, we have $x^*(\theta) = [x_1, x_2] \subset (\theta^*, c(1)]$. Moreover, $F$-almost every $x \in [x_1, x_2]$ is matched, in addition to type $\theta$, with type $c^{-1}(x) \in (\theta^*, 1]$. Then we have:
    \[
        \int_0^1\int_{x_1}^{x_2}(x - \theta)\,d\pi^*(\theta,x) = \int_{c^{-1}(x_1)}^{c^{-1}(x_2)}(c(\theta) - \theta)\,df(\theta) < 0,
    \]
    which violates the martingale condition~(\ref{martingale}).

    \medskip

    \noindent \textbf{Case 2: $x^*(\theta) \cap x^*(\theta') \neq \emptyset$ for a positive $F$-measure subset of $[\uwave{\theta}, \theta^*)$.}  
    If this is the case, since $x^*$ is a singleton and the monotonicity of $x^*$ in the strong set order, there exists an interval $[\theta_1, \theta_2] \subset [\uwave{\theta}, \theta^*)$ such that $x^*(\theta) = x^*(\theta') = x > \theta^*$ for $F$-almost every $\theta, \theta' \in [\theta_1, \theta_2]$. Moreover, by Lemma~\ref{lem:alpha_pool_high}, the posterior mean $x$ is matched, in addition to types in $[\theta_1, \theta_2]$, with type $c^{-1}(x) \in (\theta^*, 1]$. Thus, the posterior mean $x$ has positive measure equal to $F(\theta_2) - F(\theta_1)$, and:
    \[
        \int_{\theta_1}^{\theta_2}(x - \theta)\,d\pi^*(\theta,x) = \int_{\theta_1}^{\theta_2}(x - \theta)\,df(\theta) > 0,
    \]
    which again violates the martingale condition~(\ref{martingale}).
\end{proof}

\begin{lemma}\label{lem:alpha_negative_assortative}
    Let $g_2:[\underline{\theta}, \theta^*] \to [\theta^*, c(1)]$ be the unique strictly decreasing and bijection defined in Lemma~\ref{lem:bijective-function}. Then, under the optimal joint distribution $\pi^*$, we have $\uwave{\theta} = \underline{\theta}$ and $x^*(\theta) := g_2(\theta)$ for $F$-almost every $\theta \in [\underline{\theta}, \theta^*)$.
\end{lemma}

\begin{proof}
    By Lemmas~\ref{lem:alpha_pool_high},~\ref{lem:alpha_pooling_interval}, and~\ref{lem:alpha_singleton_decreasing}, the mapping $x^* : [\uwave{\theta}, \theta^*) \to (\theta^*, c(1)]$ is $F$-almost everywhere strictly decreasing and is a bijection onto $(\theta^*,c(1)]$. 
    
    Moreover, by Lemma~\ref{lem:alpha_pool_high}, we have $x^*(\theta) = c(\theta)$ for $F$-almost every $\theta \in (\theta^*, 1]$.

    To satisfy the martingale condition~(\ref{martingale}), for every interval $[x_1, x_2] \subseteq [\theta^*, c(1)]$, we require:
    \begin{align*}
        \int_0^1\int_{x_1}^{x_2}(x - \theta)\,d\pi^*(\theta, x) 
        &= \int_{\uwave{\theta}}^{\theta^*} \int_{x_1}^{x_2}(x^*(\theta) - \theta) f(\theta)\,d\theta 
        + \int_{\theta^*}^{1} \int_{x_1}^{x_2}(c(\theta) - \theta) f(\theta)\,d\theta \\
        &= \int_{x^{*-1}(x_1)}^{x^{*-1}(x_2)}(x^*(\theta) - \theta) f(\theta)\,d\theta 
        + \int_{c^{-1}(x_1)}^{c^{-1}(x_2)}(c(\theta) - \theta) f(\theta)\,d\theta \\
        &= \int_{x_1}^{x_2} \left[ (x^{*-1}(x) - x) f(x^{*-1}(x)) \frac{\partial x^{*-1}(x)}{\partial x} 
        + (x - c^{-1}(x)) f(c^{-1}(x))  \frac{\partial c^{-1}(x)}{\partial x} \right] dx \\
        &= \int_{x_1}^{x_2} \left[ \frac{\partial x^{*-1}(x)}{\partial x} 
        \cdot \frac{\frac{\partial c^{-1}(x)}{\partial x} f(c^{-1}(x))}{f(x^{*-1}(x))} 
        \cdot \frac{c^{-1}(x) - x}{x^{*-1}(x) - x} \right] 
        (x^{*-1}(x) - x) f(x^{*-1}(x))\,dx \\
        &= \int_{x_1}^{x_2} \left[ \frac{\partial x^{*-1}(x)}{\partial x} 
        - G(x, x^{*-1}(x)) \right] (x^{*-1}(x) - x) f(x^{*-1}(x))\,dx \\
        &= 0,
    \end{align*}
    where the third equality follows by the change of variables $x = x^*(\theta)$ and $x = c(\theta)$ in each term.

    Recall that $g_2$ is the unique function satisfying the differential equation (Lemma~\ref{lem:bijective-function}):
    \[
        \frac{\partial g_2^{-1}(x)}{\partial x} = G(x, g_2^{-1}(x)).
    \]
    Therefore, we conclude that $x^*(\theta) = g_2(\theta)$ for $F$-almost every $\theta \in [\uwave{\theta}, \theta^*)$, and hence $\uwave{\theta} = \underline{\theta}$.
\end{proof}

To conclude, observe that Lemma~\ref{lem:alpha_pooling_interval} implies that, under $\pi^*$, $F$-almost every $\theta \in [0, \underline{\theta})$ is fully revealed. By Lemma~\ref{lem:alpha_negative_assortative}, $F$-almost every $\theta \in [\underline{\theta}, \theta^*)$ is uniquely matched with the posterior mean $g(\theta)$. Finally, by Lemma~\ref{lem:alpha_pool_high}, $F$-almost every $\theta \in (\theta^*, 1]$ is exclusively matched with the posterior mean $c(\theta)$. 

Therefore, $\pi^*_{\theta} = \pi^{\text{NAM}}_{\theta}$ for $F$-almost every $\theta \in \Theta$.

\paragraph{Step 2.} Recall that $g_2$ denotes the unique strictly decreasing, differentiable bijection defined in Lemma~\ref{lem:bijective-function}, and let $a(x):= g_2^{-1}(x)$. We will use $a$ to construct the functions $(w, q, m)$, which we will then show satisfy condition~\ref{ZP}, and such that the value of~\ref{dual} under these functions equals the value of~\ref{primal} under $\pi^{\text{NAM}}$.

To build intuition for constructing $(w,q,m)$, recall from Step~1 that, if strong duality holds and $\pi^*$ solves the primal problem~(\ref{primal}), then $x^*(\theta)=g_2(\theta)$ for $F$-almost every $\theta\in[\underline{\theta},\theta^*)$. By complementary slackness, for such $\theta$ we must have
\[
g_2(\theta)\in\arg\max_{x\in[\theta^*,\,c(1)]}\{\alpha(\theta)+q(x)(x-\theta)\}.
\]
Whenever $g_2(\theta)\in(\theta^*,c(1))$ and $q$ is differentiable, the first-order condition yields:
\[
q'(g_2(\theta))\bigl(g_2(\theta)-\theta\bigr)+q\bigl(g_2(\theta)\bigr)=0.
\]
Reparametrizing with $x=g_2(\theta)$ gives the differential equation:
\[
q'(x)\bigl(x-a(x)\bigr)+q(x)=0,\qquad x\in(\theta^*,c(1)).
\]
Solving the differential equation yields:
\[
q(x)=C\exp\left(-\int_{\theta^*}^{x}\frac{ds}{\,s-a(s)\,}\right),
\]
for some constant $C$. This furnishes the guiding form for $q$; we now proceed with the formal construction.

Define $(w, q, m)$ as follows:
\[
    q(x) = C\exp\left(-\int_{\theta^*}^{\min\{\max\{x,\theta^*\},c(1)\}}\frac{ds}{\,s-a(s)\,}\right), \quad \text{where} \quad C= \frac{-\alpha(\underline{\theta})}{c(1)-\underline{\theta}}\exp\left(-\int_{\theta^*}^{c(1)}\frac{ds}{\,s-a(s)\,}\right).
\]
Since $\alpha(\theta)>0$ for all $\theta$ it follows that $C < 0$ and hence $q(x) < 0$ for all $x \in [0, 1]$. Define:
\[
    m(x) := q(x)(1 - x),
\]
and:
\[
    w(\theta) =
    \begin{cases}
        0 & \text{if } \theta \in [0, \underline{\theta}), \\
        \alpha(\theta) + q(g_2(\theta))(g(\theta) - \theta) & \text{if } \theta \in [\underline{\theta}, \theta^*], \\
         \alpha(\theta) + q(c(\theta))(c(\theta) - \theta) & \text{if } \theta \in (\theta^*, 1].
    \end{cases}
\]

We now show that $(w, q, m)$ satisfy condition~\ref{ZP}. For this purpose, define:
\[
    y_\theta(x) := \alpha(\theta)\,\mathbf{1}_{\{x \geq \theta^*\}} + q(x)(x - \theta),
\]
so that for all $x \in [\theta^*, c(1)]$:
\[
    y'_\theta(x) := \frac{\partial y_\theta(x)}{\partial x} = \frac{\theta - a(x)}{x - a(x)}q(x).
\]
To prove that~\ref{ZP} holds, it suffices to show that for all $\theta$, the following condition is satisfied:
\begin{equation}
    w(\theta) = \max_{x \in X} \, y_\theta(x) + m(x)\,\mathbf{1}_{\{\hat{c}(\theta) > x\}}. \tag{ZP'} \label{ZP_alpha}
\end{equation}

\noindent \textbf{Case 1: $\theta \in [\underline{\theta}, \theta^*]$.} In this range, $y'_\theta(x) > 0$ for all $x \in [\theta^*, g_2(\theta))$, $y'_\theta(g_2(\theta)) = 0$, and $y'_\theta(x) < 0$ for all $x \in (g_2(\theta), c(1)]$. This is because $\theta = a(g(\theta))$, and since $a$ is strictly decreasing, $a(x) < \theta$ for $x > g_2(\theta)$, and $a(x) > \theta$ for $x < g_2(\theta)$. Moreover, $x - a(x) > 0$ and $q(x) < 0$. Therefore, $y_\theta$ is uniquely maximized at $g_2(\theta)$ over the interval $[\theta^*,c(1)]$.

Next, observe that for any $\theta \in [0,\theta^*]$:
\[
    y_{\theta}(c(1)) = \alpha(\theta) - \frac{\alpha(\underline{\theta})}{c(1)-\underline{\theta}}(c(1)-\theta).
\]
In particular, $y_{\theta}(c(1))>0$ if and only if $\frac{\alpha(\theta)}{c(1)-\theta} > \frac{\alpha(\underline{\theta})}{c(1)-\underline{\theta}}$. Since, by assumption, the mapping $\theta \mapsto \frac{\alpha(\theta)}{c(1)-\theta}$ is strictly increasing on $[0,\theta^*]$, it follows that $y_{\theta}(c(1))<0$ if $\theta \in [0,\underline{\theta})$, $y_{\underline{\theta}}(c(1))=0$, and $y_{\theta}(c(1))>0$ if $\theta \in (\underline{\theta},\theta^*]$.

Therefore, for all $x \in [0,\theta)$:
\[
    y_\theta(g_2(\theta)) \geq y_\theta(c(1)) \geq 0 
    > q(x)(x - \theta) + m(x) = C(1 - \theta)=y_{\theta}(x),
\]
where the strict inequality follows from $C<0$ and $1-\theta>0$.

Similarly, for all $x \in [\theta,\theta^*)$:
\[
    y_\theta(g_2(\theta)) \geq 0 \geq C(x - \theta)=y_{\theta}(x).
\]
where the last inequality follows from $C<0$ and $x-\theta\geq0$.

Finally, for any $x \in (c(1),1]$:
\[
    y_\theta(g_2(\theta)) \geq y_{\theta}(c(1))
    > \alpha(\theta) + q(c(1))(x-\theta)
    = y_{\theta}(x),
\]
where the strict inequality follows from $q(c(1))<0$ and $x-\theta>c(1)-\theta$.

Hence, $w(\theta) = y_\theta(g_2(\theta))$ satisfies~\ref{ZP_alpha} for all $\theta \in [\underline{\theta}, \theta^*]$.

\medskip

\noindent \textbf{Case 2: $\theta \in (\theta^*, 1]$.} Here, $y'_\theta(x) < 0$ for all $x \in [c(\theta), c(1)]$, since $\theta - a(x) > 0$ (as $a(x) \leq \theta^*$). So $y_\theta(x)$ is uniquely maximized at $c(\theta)$ over the interval $[c(\theta), c(1)]$.

Moreover, for all $x \in [0, c(\theta))$:
\begin{align*}
    y_\theta(c(\theta)) &= \alpha(\theta) + q(c(\theta))(c(\theta) - \theta)>\alpha(\theta)\\
    &> \alpha(\theta)\,\mathbf{1}_{\{x \geq \theta^*\}} + q(x)(x - \theta) + m(x) \\
    &= \alpha(\theta)\,\mathbf{1}_{\{x \geq \theta^*\}} + q(x)(1 - \theta)=y_{\theta}(x).
\end{align*}

For $x \in (c(1), 1]$, we similarly have:
\[
    y_\theta(c(\theta))\geq y_\theta(c(1)) > \alpha(\theta)+q(c(1))(x-\theta)=y_{\theta}(x).
\]

Therefore, $w(\theta) = y_\theta(c(\theta))$ satisfies~\ref{ZP_alpha} for all $\theta \in (\theta^*, 1]$.

\medskip

\noindent \textbf{Case 3: $\theta \in [0, \underline{\theta})$.} In this case, $y'_\theta(x) > 0$ for all $x \in [\theta^*, c(1)]$ since $\theta - a(x) < 0$ (as $a(x) \geq \underline{\theta}$). Thus, $y_\theta(x)$ is uniquely maximized at $c(1)$ in the interval $[\theta^*, c(1)]$. Additionally, we already showed that for any $\theta \in [0, \underline{\theta})$ we have that $y_\theta(c(1)) < 0$.

For all $x \in [0, \theta)$, $y_{\theta}(x)=q(x)(x - \theta) + m(x) = C(1 - \theta) < 0$, and for $x \in [\theta, \theta^*)$, $y_{\theta}(x)= C(x - \theta) \leq 0$.

Finally, for all $x \in (c(1),1]$:
\[
    0>y_{\theta}(c(\theta))> \alpha(\theta)+q(c(1))(x-\theta)=y_{\theta}(x).
\]

Therefore, $w(\theta) = 0$ satisfies~\ref{ZP_alpha} for all $\theta \in [0, \underline{\theta})$.

We now verify that the values of the primal and dual problems coincide. The value of~\ref{primal} under $\pi^{\text{NAM}}$ is:
\[
    V^* = \int_{0}^1 \int_{0}^1 \alpha(\theta)\,\mathbf{1}_{\{x \geq \theta^*\}}\,d\pi^{NAM}(\theta, x) 
    = \int_{\underline{\theta}}^{1} \alpha(\theta)\,d\theta.
\]
The value of~\ref{dual} under $(w, q, m)$ is:
\begin{align*}
    \int_0^1 w(\theta)f(\theta)\,d\theta &= V^* 
    + \int_{\underline{\theta}}^{\theta^*} q(g_2(\theta))(g_2(\theta) - \theta)f(\theta)\,d\theta 
    + \int_{\theta^*}^{1} q(c(\theta))(c(\theta) - \theta)f(\theta)\,d\theta \\
    &= V^* + \int_{\theta^*}^{c(1)} q(x) \left[(g_2^{-1}(x) - x)f(g_2^{-1}(x))\frac{\partial g_2^{-1}(x)}{\partial x}\,dx + (x - c^{-1}(x))f(c^{-1}(x))\frac{\partial c^{-1}(x)}{\partial x}\right]\,dx \\
    &=V^*+\int_{\theta^*}^{c(1)} q(x) \left[ \frac{\partial g_2^{-1}(x)}{\partial x} - \frac{\frac{\partial c^{-1}(x)}{\partial x}f(c^{-1}(x))}{f(g_2^{-1}(x))}\frac{c^{-1}(x) - x}{g_2^{-1}(x) - x} \right] (g_2^{-1}(x) - x)f(g^{-1}(x))\,dx \\
    &=V^*+\int_{\theta^*}^{c(1)} q(x)\left[\frac{\partial g_2^{-1}(x)}{\partial x} - G(x, g_2^{-1}(x))\right](g_2^{-1}(x) - x)f(g_2^{-1}(x))\,dx\\
    &=V^*.
\end{align*}
The second equality follows from the change of variables $x = g_2(\theta)$ and $x = c(\theta)$, and the last equality follows from the fact that $g_2$ satisfies the differential equation (Lemma~\ref{lem:bijective-function}).

Hence, by Lemma~\ref{cor:strong-duality}, $\pi^{\text{NAM}}$ solves~\ref{primal}, $(w, q, m)$ solves~\ref{dual}, and strong duality holds.

\subsection{Proof of Proposition \ref{prop:alpha_convex}}

\paragraph{Step 1.} Suppose that strong duality holds. Let $\pi^*$ be any solution to the primal problem~(\ref{primal}), and let $(w^*,q^*,m^*)$ be any solution to the dual problem~(\ref{dual}).

We follow the reasoning of Step 1 in the proof of Proposition~\ref{prop:alpha_increasing} and rely on some of the lemmas established there (Section~\ref{sec:proof_alpha_increasing}) to show that there exists $x^* \in [\theta^*,c(1)]$ such that $\pi^*_{\theta}=\pi^{x^*-\text{NAM}}$ for $F$-almost every $\theta$.

Observe that in the proofs of Lemmas~\ref{lem:q_negative} and~\ref{lem:alpha_pool_high} we did not rely on any assumption regarding the behavior of the function $\theta \mapsto \frac{\alpha(\theta)}{x-\theta}$. Therefore, these lemmas remain valid and imply that $q^*(x)<0$ for all $x \in [\theta^*,1]$ and that, under the optimal joint distribution $\pi^*$, $F$-almost every $\theta \in (\theta^*,1]$ is pooled with some type $\theta'$ so as to achieve a posterior mean equal to its cost, $x^*(\theta)=c(\theta)$. Combining these facts, we obtain the following new result:

\begin{lemma}\label{lem:alpha_disjoint_interval}
    There exist types $0<\theta_1<\theta_2<\theta^*$ such that under the optimal joint distribution $\pi^*$, $F$-almost every $\theta \in (\theta_1, \theta_2)$ is fully revealed, and $F$-almost every $\theta \in [0,\theta_1] \cup [\theta_2,\theta^*)$ is pooled (with probability one).
\end{lemma}
\begin{proof}
    Let $A \subset [0,\theta^*)$ denote the set of pooled types. By Lemma~\ref{lem:alpha_pool_high}, $F(A)>0$. For each $x \in [\theta^*,c(1)]$, define $t:[\theta^*,c(1)] \times [0,\theta^*] \to \mathbb{R}$ and write $t_x(\theta):=t(x,\theta)$. By assumption, each $t_x$ is strictly convex and therefore admits a unique minimizer $\theta_x \in [0,\theta^*]$.

    For $\theta \in A$, let $x(\theta) \ge \theta^*$ denote the posterior mean assigned to $\theta$ under the optimal joint distribution $\pi^*$. By complementary slackness, for $F$-almost every $\theta \in A$:
    \[
        \alpha(\theta) + q^*(x(\theta))\bigl(x(\theta) - \theta\bigr) \ge 0
        \quad\Rightarrow\quad
        t_{x(\theta)}(\theta)\;=\;\frac{\alpha(\theta)}{x(\theta)-\theta}\;\ge\;-\,q^*(x(\theta)).
    \]

    It cannot be that $\theta=\theta_{x(\theta)}$ for $F$-almost every $\theta \in A$. Indeed, if $\theta=\theta_{x(\theta)}$, then $t_{x(\theta)}(\theta')>t_{x(\theta)}(\theta)\ge -q^*(x(\theta))$ for all $\theta' \in [0,\theta^*]$, which implies:
    \[
        \alpha(\theta') + q^*(x(\theta))\bigl(x(\theta)-\theta'\bigr) > 0,
    \]
    for all $\theta' \in [0,\theta^*]$. Hence almost every type below $\theta^*$ would be pooled, contradicting Assumption~\ref{assn:no_full_trade}. Consequently, a positive $F$-measure of types in $A$ must satisfy either $\theta<\theta_{x(\theta)}$ or $\theta>\theta_{x(\theta)}$.

     \medskip

    \noindent\textbf{Case 1: $\;F(\{\theta\in A:\theta<\theta_{x(\theta)}\})>0$.}
    Let $S^-:=\{\theta\in A:\theta<\theta_{x(\theta)}\}$ and set $\theta_1:=\sup S^-$. For any $\theta\in S^-$ and any $\theta'<\theta$, strict convexity yields $t_{x(\theta)}(\theta')>t_{x(\theta)}(\theta)\ge -q^*(x(\theta))$, so $\alpha(\theta')+q^*(x(\theta))\bigl(x(\theta)-\theta'\bigr)>0$, i.e., $\theta'$ is pooled. Taking limits along $S^-$ shows that $F$-almost every type in $[0,\theta_1]$ is pooled. By Assumption~\ref{assn:no_full_trade}, we must have $\theta_1<\theta^*$.

    We now show that there exists $\theta_2<\theta^*$ such that $F$-almost every type in $[\theta_2,\theta^*)$ is pooled. Suppose not. Then, by Lemma~\ref{lem:alpha_pool_high}, any $x\in(\theta^*,c(1)]$ must be matched with some $\theta \in [0,\theta_1]$. However, since $\alpha$ is strictly positive and bounded there is $0 <M<\infty$ such that $\lim_{x\to\theta^*} t_x(\theta)<M$ for all $\theta\in[0,\theta_1]$, while $\lim_{x\to\theta^*} t_x(\theta^*)=\infty$. By continuity of $t_x$, there exists $\varepsilon>0$ and $\theta_2<\theta^*$ such that, for all $x\in(\theta^*,\theta^*+\varepsilon]$ and all $\theta'\in[\theta_2,\theta^*)$, $t_x(\theta')>t_x(\theta)$ for every $\theta\in[0,\theta_1]$.
    
    If some $\theta\in[0,\theta_1]$ is matched with such an $x$, then $t_x(\theta)\ge -q^*(x)$, hence $t_x(\theta')>-q^*(x)$, which makes matching $\theta'$ with $x$ strictly more valuable than revealing $\theta'$. Thus $[\theta_2,\theta^*)$ must be pooled $F$-almost everywhere, a contradiction to the supposition.

    \medskip
    
    \noindent \textbf{Case 2: $\;F(\{\theta\in A:\theta>\theta_{x(\theta)}\})>0$.}
    Let $S^+:=\{\theta\in A:\theta>\theta_{x(\theta)}\}$ and set $\theta_2:=\inf S^+$. For any $\theta\in S^+$ and any $\theta'>\theta$, strict convexity implies $t_{x(\theta)}(\theta')>t_{x(\theta)}(\theta)\ge -q^*(x(\theta))$, so $\alpha(\theta')+q^*(x(\theta))\bigl(x(\theta)-\theta'\bigr)>0$, i.e., $\theta'$ is pooled. Hence $F$-almost every type in $[\theta_2,\theta^*)$ is pooled. By Assumption~\ref{assn:no_full_trade}, we must have $\theta_2>0$.

    We next show that there exists $\theta_1>0$ such that $F$-almost every type in $[0,\theta_1]$ is pooled. Suppose not. By Lemma~\ref{lem:alpha_pool_high}, any $x\in(\theta^*,c(1)]$ must then be matched with some $\theta\in[\theta_2,\theta^*)$. Using Lemmas~\ref{lem:alpha_singleton_decreasing} and~\ref{lem:alpha_negative_assortative}, it follows that $\theta_2=\underline{\theta}$ and $x(\theta)=g_2(\theta)$ for $\theta\in[\underline{\theta},\theta^*)$, where $g_2$ is the strictly decreasing bijection defined in Lemma~\ref{lem:bijective-function}. Moreover, $x(\underline{\theta})=g_2(\underline{\theta})=c(1)$ and, by assumption, $t_{c(1)}(0)>t_{c(1)}(\underline{\theta})$. By continuity, there exist $\varepsilon>0$ and $\theta_1>0$ such that, for all $\theta'\in[0,\theta_1]$ and all $\theta\in[\underline{\theta},\underline{\theta}+\varepsilon]$, $t_{g(\theta)}(\theta') \;>\; t_{g(\theta)}(\theta)$.
    
    If $\theta$ is matched with $g(\theta)$, then $t_{g(\theta)}(\theta)\ge -q^*(g(\theta))$, hence $t_{g(\theta)}(\theta')>-q^*(g(\theta))$, which makes matching $\theta'$ with $g(\theta)$ strictly more valuable than revealing $\theta'$. Thus $[0,\theta_1]$ must be pooled $F$-almost everywhere, a contradiction to the supposition.

    \medskip
    Combining the two cases, we conclude that there exist $0<\theta_1<\theta_2<\theta^*$ such that $A=[0,\theta_1] \,\cup\, [\theta_2,\theta^*)$,
    up to an $F$-null set.
\end{proof}

Once we have shown Lemma~\ref{lem:alpha_disjoint_interval}, we can apply the same argument as in Lemma~\ref{lem:alpha_singleton_decreasing} to conclude that, under the optimal joint distribution $\pi^*$, $x^*(\theta)$ takes singleton values and is strictly decreasing for $F$-almost every $\theta \in [0,\theta_1] \cup [\theta_2,\theta^*)$. We can then establish the following new result:

\begin{lemma}\label{lem:alpha_truncated_negative_assortative}
    Let $g_1:[0,\bar{\theta}]\to[\theta^*,c(1)]$ and $g_2:[\underline{\theta},\theta^*]\to[\theta^*,c(1)]$ denote the unique strictly decreasing bijections defined in Lemma~\ref{lem:bijective-function}. Under the optimal joint distribution $\pi^*$, up to an $F$-null set we have $x^*(\theta)=g_1(\theta)$ on $[0,\theta_1]$ and $x^*(\theta)=g_2(\theta)$ on $[\theta_2,\theta^*)$, and moreover $x^*(\theta_1)=x^*(\theta_2)$.
\end{lemma}

\begin{proof}
  By Lemmas~\ref{lem:alpha_pool_high}, \ref{lem:alpha_disjoint_interval}, and~\ref{lem:alpha_singleton_decreasing}, the assignment $x^*:[0,\theta_1]\,\cup\,[\theta_2,\theta^*) \to (\theta^*,c(1)]$ is strictly decreasing $F$-almost everywhere and is a bijection onto $(\theta^*,c(1)]$. Consequently, the images $x^*([0,\theta_1])$ and $x^*([\theta_2,\theta^*))$ are adjacent intervals whose union is $(\theta^*,c(1)]$, and therefore their boundary values coincide $x^*(\theta_1)=x^*(\theta_2)$.

    Additionally, by Lemma~\ref{lem:alpha_pool_high}, $x^*(\theta)=c(\theta)$ for $F$-almost every $\theta\in(\theta^*,1]$.
    
    To satisfy the martingale condition~(\ref{martingale}), for every interval $[x_1, x_2] \subseteq [c(1),x^*]$, we require:
    \begin{align*}
        \int_0^1\int_{x_1}^{x_2}(x - \theta)\,d\pi^*(\theta, x) 
        &= \int_{0}^{\theta^*} \int_{x_1}^{x_2}(x^*(\theta) - \theta) f(\theta)\,d\theta 
        + \int_{\theta^*}^{1} \int_{x_1}^{x_2}(c(\theta) - \theta) f(\theta)\,d\theta \\
        &= \int_{x^{*-1}(x_1)}^{x^{*-1}(x_2)}(x^*(\theta) - \theta) f(\theta)\,d\theta 
        + \int_{c^{-1}(x_1)}^{c^{-1}(x_2)}(c(\theta) - \theta) f(\theta)\,d\theta \\
        &= \int_{x_1}^{x_2} \left[ (x^{*-1}(x) - x) f(x^{*-1}(x)) \frac{\partial x^{*-1}(x)}{\partial x} 
        + (x - c^{-1}(x)) f(c^{-1}(x))  \frac{\partial c^{-1}(x)}{\partial x} \right] dx \\
        &= \int_{x_1}^{x_2} \left[ \frac{\partial x^{*-1}(x)}{\partial x} 
        \cdot \frac{\frac{\partial c^{-1}(x)}{\partial x} f(c^{-1}(x))}{f(x^{*-1}(x))} 
        \cdot \frac{c^{-1}(x) - x}{x^{*-1}(x) - x} \right] 
        (x^{*-1}(x) - x) f(x^{*-1}(x))\,dx \\
        &= \int_{x_1}^{x_2} \left[ \frac{\partial x^{*-1}(x)}{\partial x} 
        - G(x, x^{*-1}(x)) \right] (x^{*-1}(x) - x) f(x^{*-1}(x))\,dx \\
        &= 0,
    \end{align*}
    where the third equality follows by the change of variables $x = x^*(\theta)$ and $x = c(\theta)$ in each term.

    Recall that $g_1:[0,\bar{\theta}]\to[\theta^*,c(1)]$ is the strictly decreasing bijection from Lemma~\ref{lem:bijective-function}, and let $b:=h^{-1}:[\theta^*,c(1)]\to[0,\bar{\theta}]$ denote its inverse. By Lemma~\ref{lem:bijective-function}, $b$ is the unique solution to the partial differential equation given by $b'(x)=G(x,b(x))$.

    By uniqueness together with $b(h(\theta))=\theta$, it follows that $x^*(\theta)=g_1(\theta)$ for $F$-almost every $\theta\in[0,\theta_1]$. An analogous argument—applied to $g_2$ (on any interval $[x_1,x_2]\subseteq[x^*,c(1)]$) yields $x^*(\theta)=g_2(\theta)$ for $F$-almost every $\theta\in[\theta_2,\theta^*)$.
\end{proof}

\paragraph{Step 2.} Recall that $g_1$ and $g_2$ denote the functions defined in Lemma~\ref{lem:bijective-function}, and let $a(x) := g^{-1}_2(x)$ and $b(x) := g_1^{-1}(x)$. 

Define $t:[\theta^*,c(1)] \times [0,\theta^*] \to \mathbb{R}$ by $t_x(\theta):=t(x,\theta) = \frac{\alpha(\theta)}{x-\theta}$. The mapping $x \mapsto t_x(b(x)) - t_x(a(x))$ is continuous and, by assumption, is strictly positive at $x = c(1)$ because $\frac{\alpha(0)}{c(1)} - \frac{\alpha(\underline{\theta})}{c(1)-\underline{\theta}} > 0$. Moreover, it is strictly negative for $x$ sufficiently close to $\theta^*$, since $\lim_{x \to \theta^*} t_x(a(x)) = -\infty$.

By the intermediate value theorem, there exists $x^* \in (0,\theta^*)$ such that $t_{x^*}(b(x^*)) - t_{x^*}(a(x^*)) = 0$. Let $b(x^*) = \theta_1$ and $a(x^*) = \theta_2$, where $0 < \theta_1 < \theta_2 < \theta^*$. The assumption that $t_x$ is strictly convex for each $x \in [\theta^*,c(1)]$ then implies that $t_{x^*}(\theta) < t_{x^*}(\theta_1) = t_{x^*}(\theta_2)$ for all $\theta \in (\theta_1,\theta_2)$, and $t_{x^*}(\theta) > t_{x^*}(\theta_1) = t_{x^*}(\theta_2)$ for all $\theta \in [0,\theta_1) \cup (\theta_2,\theta^*]$.

We use $a$, $b$, $\theta_1$, $\theta_2$, and $x^*$ to construct the functions $(w,q,m)$. We then show that these functions satisfy condition~\ref{ZP} and that the value of~\ref{dual} under them equals the value of~\ref{primal} under $\pi^{x^*-\text{NAM}}$. For intuition on this construction, refer to Step~2 of the proof of Proposition \ref{prop:alpha_increasing} (Section \ref{sec:proof_alpha_increasing}).

Define $(w, q, m)$ as follows:
\[
    q(x) =\begin{cases}
        C\exp\left({-\int_{x^*}^{\max\{x,\theta^*\}}\frac{ds}{s-a(s)}}\right) & \text{if } x \in [0, x^*], \\
        C\exp\left({-\int_{x^*}^{\min\{x,c(1)\}}\frac{ds}{s-b(s)}}\right) & \text{if } x \in [x^*,1],
    \end{cases} 
\]
where: 
\[
C= \frac{-\alpha(\theta_1)}{x^*-\theta_1}=\frac{-\alpha(\theta_2)}{x^*-\theta_2}.
\]
Since $\alpha(\theta)>0$ for all $\theta$ it follows that $C< 0$ and hence $q(x) < 0$ for all $x \in [0, 1]$. Define:
\[
    m(x) := q(x)(1 - x),
\]
and:
\[
    w(\theta) =
    \begin{cases}
        \alpha(\theta)+q(g_1(\theta))(g_1(\theta)-\theta) & \text{if } \theta \in [0, \theta_1], \\
        0 & \text{if } \theta \in (\theta_1, \theta_2), \\
        \alpha(\theta) + q(g_2(\theta))(g_2(\theta) - \theta) & \text{if } \theta \in [\underline{\theta}, \theta^*], \\
         \alpha(\theta) + q(c(\theta))(c(\theta) - \theta) & \text{if } \theta \in (\theta^*, 1].
    \end{cases}
\]

We now show that $(w, q, m)$ satisfy condition~\ref{ZP}. For this purpose, define:
\[
    y_\theta(x) := \alpha(\theta)\,\mathbf{1}_{\{x \geq \theta^*\}} + q(x)(x - \theta),
\]
so that for all $x \in [\theta^*, c(1)]$:
\[
    y'_\theta(x) := \frac{\partial y_\theta(x)}{\partial x} =\begin{cases} \frac{\theta - a(x)}{x - a(x)}q(x) & \text{if } x \in [\theta^*, x^*),\\
    \frac{\theta - b(x)}{x - b(x)}q(x) & \text{if } x \in (x^*,c(1)].
    \end{cases}
\]

To prove that~\ref{ZP} holds, it suffices to show that for all $\theta$, the following condition is satisfied:
\begin{equation}
    w(\theta) = \max_{x \in X} \, y_\theta(x) + m(x)\,\mathbf{1}_{\{\hat{c}(\theta) > x\}}. \tag{ZP'} \label{ZP_alpha2}
\end{equation}

\noindent \textbf{Case 1: $\theta \in [0, \theta_1]$.} In this range, $y'_\theta(x) > 0$ for all $x \in [\theta^*, x^*]$. This is because $\theta<\theta_2 \leq a(x)$, $x-a(x)>0$ and $q(x)<0$. Then, $y_\theta$ is uniquely maximized at $x^*$ over the interval $[\theta^*,x^*]$

On the other hand, $y'_\theta(x) > 0$ for all $x \in [x^*, g_1(\theta))$, $y'_\theta(g_1(\theta)) = 0$, and $y'_\theta(x) < 0$ for all $x \in (g_1(\theta), c(1)]$. This is because $\theta = b(g_1(\theta))$, and since $b$ is strictly decreasing, $b(x) < \theta$ for $x > g_1(\theta)$, and $b(x) > \theta$ for $x < b(\theta)$. Moreover, $x - b(x) > 0$ and $q(x) < 0$. Therefore, $y_\theta$ is uniquely maximized at $g_1(\theta)$ over $[x^*,c(1)]$, and, as noted in the previous paragraph, this maximum also holds over the entire interval $[\theta^*,c(1)]$.

Next, observe that for any $\theta \in [0,\theta^*]$:
\[
    y_{\theta}(x^*) = \alpha(\theta) - \frac{\alpha(\theta_1)}{x^*-\theta_1}\,(x^*-\theta)=\alpha(\theta) - t_{x^*}(\theta_1)\,(x^*-\theta)=\alpha(\theta) - t_{x^*}(\theta_2)\,(x^*-\theta).
\]
In particular, $y_{\theta}(x^*)>0$ if and only if $t_{x^*}(\theta)>t_{x^*}(\theta_1)=t_{x^*}(\theta_2)$, which we now is true if and only if $\theta \in [0,\theta_1) \cup (\theta_2,\theta^*]$. Therefore, $y_{\theta}(x^*)<0$ if $\theta \in (\theta_1,\theta_2)$, $y_{\theta_1}(x^*)=y_{\theta_2}(x^*)=0$, and $y_{\theta}(x^*)>0$ if $\theta \in [0,\theta_1) \cup (\theta_2,\theta^*]$.

Therefore, for all $x \in [0,\theta)$:
\[
    y_\theta(g_1(\theta)) \geq y_\theta(x^*) \geq 0 
    > q(x)(x - \theta) + m(x) = C(1 - \theta)=y_{\theta}(x),
\]
where the strict inequality follows from $C<0$ and $1-\theta>0$.

Similarly, for all $x \in [\theta,\theta^*)$:
\[
    y_\theta(g_1(\theta)) \geq 0 \geq C(x - \theta)=y_{\theta}(x).
\]
where the weak inequality follows from $C<0$ and $x-\theta \geq 0$.

Finally, for any $x \in (c(1),1]$:
\[
    y_\theta(g_1(\theta)) \geq y_{\theta}(c(1))
    > \alpha(\theta) + q(c(1))(x-\theta)
    = y_{\theta}(x),
\]
where the strict inequality follows from $q(c(1))<0$ and $x-\theta>c(1)-\theta$.

Hence, $w(\theta) = y_\theta(g_1(\theta))$ satisfies~\ref{ZP_alpha2} for all $\theta \in [0, \theta_1]$.

\noindent \textbf{Case 2: $\theta \in [\theta_2,\theta^*]$.} We omit the proof of this case, since a symmetric argument to that used in case 1 shows that $w(\theta) = y_\theta(g_2(\theta))$ satisfies~\ref{ZP_alpha2} for all $\theta \in [\theta_2, \theta^*]$.

\noindent \textbf{Case 3: $\theta \in (\theta^*, 1]$.} Here, $y'_\theta(x) < 0$ for all $x \in [c(\theta), c(1)]$, since $\theta - a(x),\theta-b(x) > 0$ (as $a(x)$, $b(x) \leq \theta^*$). So $y_\theta(x)$ is uniquely maximized at $c(\theta)$ over the interval $[c(\theta), c(1)]$.

Moreover, for all $x \in [0, c(\theta))$:
\begin{align*}
    y_\theta(c(\theta)) &= \alpha(\theta) + q(c(\theta))(c(\theta) - \theta)>\alpha(\theta)\\
    &> \alpha(\theta)\,\mathbf{1}_{\{x \geq \theta^*\}} + q(x)(x - \theta) + m(x) \\
    &= \alpha(\theta)\,\mathbf{1}_{\{x \geq \theta^*\}} + q(x)(1 - \theta)=y_{\theta}(x).
\end{align*}

For $x \in (c(1), 1]$, we similarly have:
\[
    y_\theta(c(\theta))\geq y_\theta(c(1)) > \alpha(\theta)+q(c(1))(x-\theta)=y_{\theta}(x).
\]

Therefore, $w(\theta) = y_\theta(c(\theta))$ satisfies~\ref{ZP_alpha2} for all $\theta \in (\theta^*, 1]$.

\noindent \textbf{Case 4: $\theta \in (\theta_1,\theta_2)$.} In this case, $y'_\theta(x) < 0$ for all $x \in [x^*, c(1)]$ since $\theta - b(x) \geq \theta-\theta_1> 0$. On the other hand, $y'_\theta(x) > 0$ for all $x \in [0, x^*]$ since $\theta - a(x) \leq \theta-\theta_2< 0$  Thus, $y_\theta(x)$ is uniquely maximized at $x^*$ in the interval $[\theta^*, c(1)]$.

Additionally, we already showed that for any $\theta \in (\theta_1, \theta_2)$ we have that $y_\theta(x^*) < 0$.

For all $x \in [0, \theta)$, $y_{\theta}(x)=q(x)(x - \theta) + m(x) = C(1 - \theta) < 0$, and for $x \in [\theta, \theta^*)$, $y_{\theta}(x)= C(x - \theta) \leq 0$.

Finally, for all $x \in (c(1),1]$:
\[
    0>y_{\theta}(c(\theta))> \alpha(\theta)+q(c(1))(x-\theta)=y_{\theta}(x).
\]

Therefore, $w(\theta) = 0$ satisfies~\ref{ZP_alpha2} for all $\theta \in (\theta_1, \theta_2)$.

We now verify that the values of the primal and dual problems coincide. The value of~\ref{primal} under $\pi^{x^*-\text{NAM}}$ is:
\[
    V^* = \int_{0}^1 \int_{0}^1 \alpha(\theta)\,\mathbf{1}_{\{x \geq \theta^*\}}\,d\pi^{NAM}(\theta, x) 
    = \int_{0}^{\theta_1} \alpha(\theta)\,d\theta+\int_{\theta_2}^{1} \alpha(\theta)\,d\theta.
\]

The value of~\ref{dual} under $(w, q, m)$ is:
\begin{align*}
    \int_0^1 w(\theta)f(\theta)\,d\theta =& V^* 
    + \int_{0}^{\theta_1} q(g_1(\theta))(g_1(\theta) - \theta)f(\theta)\,d\theta+ + \int_{\theta_2}^{\theta^*} q(g_2(\theta))(g_2(\theta) - \theta)f(\theta)\,d\theta\\
    &+ \int_{\theta^*}^{1} q(c(\theta))(c(\theta) - \theta)f(\theta)\,d\theta \\
    =& V^* + \int_{\theta^*}^{x^*} q(x) \left[(g_2^{-1}(x) - x)f(g_2^{-1}(x))\frac{\partial g_2^{-1}(x)}{\partial x}\,dx + (x - c^{-1}(x))f(c^{-1}(x))\frac{\partial c^{-1}(x)}{\partial x}\right]\,dx \\
    &+\int_{x^*}^{c(1)} q(x) \left[(g_1^{-1}(x) - x)f(g_1^{-1}(x))\frac{\partial g_1^{-1}(x)}{\partial x}\,dx + (x - c^{-1}(x))f(c^{-1}(x))\frac{\partial c^{-1}(x)}{\partial x}\right]\,dx\\
    =&V^*+\int_{\theta^*}^{x^*} q(x)\left[\frac{\partial g_2^{-1}(x)}{\partial x} - G(x, g_2^{-1}(x))\right](g_2^{-1}(x) - x)f(g_2^{-1}(x))\,dx\\
    &+\int_{x^*}^{c(1)} q(x)\left[\frac{\partial g_1^{-1}(x)}{\partial x} - G(x, g_1^{-1}(x))\right](g_1^{-1}(x) - x)f(g_1^{-1}(x))\,dx\\
    =&V^*.
\end{align*}
The second equality follows from the change of variables $x = g_1(\theta)$, $x=g_2(\theta)$ and $x = c(\theta)$, and the last equality follows from the fact that $g_1$ and $g_2$ satisfy the differential equation (Lemma~\ref{lem:bijective-function}). 

Hence, by Lemma~\ref{cor:strong-duality}, $\pi^{x^*-\text{NAM}}$ solves~\ref{primal}, $(w, q, m)$ solves~\ref{dual}, and strong duality holds.

\section{Proofs of Section \ref{sec:price_surplus}}\label{sec:proof_price_surplus}

\begin{proof}[Proof of Part 2 of Proposition \ref{prop:revenue-surplus}]
    Recall that the designer’s problem is:
    \[
        \max_{\pi \in \Delta(\Theta \times X)} 
        \int_0^1\int_0^1 [x-(1-\beta)c(\theta)] \cdot \mathbf{1}_{\{x\ge \theta^*\}} \, d\pi(\theta, x) 
        \quad \text{s.t. \ref{BP}, \ref{martingale}, and \ref{price}}.
    \]
    The martingale constraint (\ref{martingale}) implies that for all measurable $B \subseteq X$:
    \[
        \int_0^1\int_B x\,d\pi(\theta,x) = \int_0^1\int_B \theta\,d\pi(\theta,x).
    \]
    Hence, without loss of optimality, we can rewrite the objective as:
    \[
        \int_0^1\int_0^1 [\theta-(1-\beta)c(\theta)] \cdot \mathbf{1}_{\{x\ge \theta^*\}} \, d\pi(\theta, x).
    \]
    Note that $\theta-(1-\beta)c(\theta)>0$ if $\theta>\theta_{\beta}$ and $\theta-(1-\beta)c(\theta)<0$ if $\theta<\theta_{\beta}$. Therefore, an upper bound on the designer’s objective is:
    \[
        \int_{0}^1\int_0^1[\theta-(1-\beta)c(\theta)]\cdot\mathbf{1}_{\{x\geq\theta^*\}}\,d\pi(\theta,x) 
        \;\leq\; \int_{\theta_{\beta}}^1[\theta-(1-\beta)c(\theta)]f(\theta)\,d\theta,
    \]
    where the inequality holds for any $\pi$ satisfying (\ref{BP}), (\ref{martingale}), and (\ref{price}).  

    This bound is attained if and only if, under $\pi$, $f$-almost every type in $[0,\theta_{\beta}]$ does not trade, and $f$-almost every type in $[\theta_{\beta},1]$ does trade. Moreover, such a $\pi$ exists. Indeed, since $\underline{\theta} \geq \theta_{\beta}$, we can construct it by modifying $\pi^{\text{NAM}}$: rather than starting pooling at $\underline{\theta}$ via $g_2$, we start pooling at $\theta_{\beta}$. Formally:
    \[
        d\pi(\theta,x)=
        \begin{cases}
            f(\theta)\delta_{\theta}(x) \,d\theta & \text{if } \theta \in [0,\theta_{\beta}), \\
            f(\theta)\delta_{g_2(\theta)}(x) \,d\theta & \text{if } \theta \in [\theta_{\beta},\theta^*], \\
            f(\theta)\delta_{c(\theta)}(x) \,d\theta & \text{if } \theta \in (\theta^*,1].
        \end{cases}
    \]
\end{proof}

\end{document}